\documentclass[preprint,12pt]{elsarticle}

\usepackage{array}
\usepackage{epsf}
\usepackage[intlimits]{amsmath}
\usepackage{amsthm}
\usepackage{amssymb}
\usepackage{subfigure}
\usepackage{setspace}
\usepackage{float}
\usepackage{fullpage}
\usepackage{algorithmic}
\usepackage{algorithm}
\usepackage{color}
\usepackage{hyperref}
\usepackage{graphicx}
\usepackage{subfloat}
\usepackage{bigints}

\DeclareMathAlphabet\mathbfcal{OMS}{cmsy}{b}{n}

\newtheorem{theorem}{Theorem}[section]
\newtheorem{corollary}{Corollary}[section]
\newtheorem{definition}{Definition}[section]
\newtheorem{remark}{Remark}[section]

\journal{Journal of Computational Physics}

\begin{document}
%
%
%

\newcommand{\be}{\begin{equation}}
\newcommand{\ee}{\end{equation}}

\newcount{\myi}
\newcommand{\Repeat}[2]{%
    \myi=0
    \loop
        \ifnum\myi<#2
        #1
        \advance\myi by 1
    \repeat
}

\newcommand\Sum[2]{\sum_{#1}^{#2}}



\newcommand{\ucont}{{\it u}}
\newcommand{\qcont}{{\it q}}
\newcommand{\fcont}{{\bf F}}
\newcommand{\TE}{{\it T}_e}
\newcommand{\Dmat}{\mathcal{D}}
\newcommand{\Pinv}{\mathcal{P}^{-1}}
\newcommand{\Qmat}{\mathcal{Q}}
\newcommand{\Bmat}{\mathcal{B}}

\newcommand{\xMj}{{\bar{x}}_{j}}

\newcommand{\Uc}{\hat{\mathbf{U}}}
\newcommand{\Fc}{\hat{\mathbf{F}}}
\newcommand{\Gc}{\hat{\mathbf{G}}}
\newcommand{\Hc}{\hat{\mathbf{H}}}
\newcommand{\Sc}{\hat{\mathbf{S}}}
\newcommand{\Ub}{\mathbf{U}}
\newcommand{\Fb}{\mathbf{F}}
\newcommand{\qN}{\mathbf{q}}
\newcommand{\uN}{\mathbf{u}}
\newcommand{\vN}{\mathbf{v}}
\newcommand{\wN}{\mathbf{w}}
\newcommand{\WN}{\mathbf{W}}
\newcommand{\fN}{\mathbf{f}}
\newcommand{\gN}{\mathbf{g}}
\newcommand{\SN}{\mathbf{S}}
\newcommand{\eN}{\mathbf{e}}

\newcommand{\qbh}{\bar q}
\newcommand{\DeltaH}{\bar \Delta}
\newcommand{\fbbh}{\Bar{\Bar{f}}}
\newcommand{\Fbbh}{\Bar{\Bar{F}}}
\newcommand{\fbb}{\mathbf{\Bar{\Bar{f}}}}
\newcommand{\Fbb}{\mathbf{\Bar{\Bar{F}}}}
\newcommand{\fb}{\mathbf{\overline{f}}}
\newcommand{\fM}{\mathbf{\overline{f}}}
\newcommand{\FM}{\mathbf{\overline{F}}}
\newcommand{\FN}{\mathbf{{F}}}
\newcommand{\uM}{\mathbf{\overline{u}}}
\newcommand{\xN}{\mathbf{x}}
\newcommand{\xiN}{\boldsymbol{\xi}}
\newcommand{\xiM}{\bar{\boldsymbol{\xi}}}
\newcommand{\xih}{\hat{\xi}}
\newcommand{\phiN}{\boldsymbol{\phi}}
\newcommand{\psiN}{\boldsymbol{\psi}}
\newcommand{\vthetaN}{\boldsymbol{\vartheta}}
\newcommand{\xb}{\mathbf{\bar{x}}}
\newcommand{\xM}{\mathbf{\bar{x}}}
\newcommand{\dX}{\Delta x}
\newcommand{\pinv}{\mathcal{P}^{-1}}
\newcommand{\ulpinv}{\underline{\mathcal{P}}^{-1}}
\newcommand{\ullpinv}{\underline{\underline{\mathcal{P}}}^{-1}}
\newcommand{\dmat}{\mathcal{D}}
\newcommand{\dtwo}{\mathcal{D}_{2}}
\newcommand{\bmat}{\mathcal{B}}
\newcommand{\cmat}{\mathcal{C}}
\newcommand{\chatmat}{\widehat{c}}
\newcommand{\gmat}{\mathcal{G}}
\newcommand{\pmath}{\hat{\mathcal{P}}}
\newcommand{\phinv}{\hat{\mathcal{P}}^{-1}}
\newcommand{\pmat}{\mathcal{P}}
\newcommand{\Pmat}{\mathcal{P}}
\newcommand{\pmatv}{\boldsymbol{\mathcal{P}}}
\newcommand{\qmat}{\mathcal{Q}}
\newcommand{\rmat}{\mathcal{R}}
\newcommand{\Inmat}{\mathcal{I}}
\newcommand{\Imat}{\bar{\mathcal{I}}}
\newcommand{\Jmat}{\bar{\mathcal{J}}}
\newcommand{\tmat}{\mathcal{T}}
\newcommand{\umat}{\mathcal{U}}
\newcommand{\vmat}{\mathcal{V}}
\newcommand{\wmat}{\mathcal{W}}
\newcommand{\Oh}{\mathcal{O}}
\newcommand{\rmatb}{\bar{\mathcal{R}}}
\newcommand{\smat}{\mathcal{S}}
\newcommand{\mmat}{\mathcal{M}}
\newcommand{\nmat}{\mathcal{N}}
\newcommand{\Gb}{\mathbf{G}}
\newcommand{\Hb}{\mathbf{H}}
\newcommand{\gb}{\mathbf{g}}
\newcommand{\vt}{\tidle{v}}
\newcommand{\ddt}{\frac{d}{dt}}
\newcommand{\gmog}{\frac{\gamma_0-1}{\gamma_0}}
\newcommand{\csqMi}{\frac{c}{\sqrt{M_0}}}
\newcommand{\csqM}{c \sqrt{M_0}}
\newcommand{\ltmpxi}{\sqrt{\xi_x^2 + \xi_y^2 + \xi_z^2}}
\newcommand{\ltmpeta}{\sqrt{\eta_x^2 + \eta_y^2 + \eta_z^2}}
\newcommand{\ltmpzeta}{\sqrt{\zeta_x^2 + \zeta_y^2 + \zeta_z^2}}
\newcommand{\ltmp}{\sqrt{{\kappa_i}_x^2 + {\kappa_i}_y^2 + {\kappa_i}_z^2}}
\newcommand{\xixh}{\hat{\xi}_x}
\newcommand{\xiyh}{\hat{\xi}_y}
\newcommand{\xizh}{\hat{\xi}_z}
\newcommand{\Dp}{D^{+}}
\newcommand{\Dm}{D^{-}}
\newcommand{\Fcp}{\hat{\mathbf{F}}^{+}}
\newcommand{\Gcp}{\hat{\mathbf{G}}^{+}}
\newcommand{\Hcp}{\hat{\mathbf{H}}^{+}}
\newcommand{\Fcm}{\hat{\mathbf{F}}^{-}}
\newcommand{\Gcm}{\hat{\mathbf{G}}^{-}}
\newcommand{\Hcm}{\hat{\mathbf{H}}^{-}}
\newcommand{\Fbp}{\mathbf{F}^{+}}
\newcommand{\Gbp}{\mathbf{G}^{+}}
\newcommand{\Hbp}{\mathbf{H}^{+}}
\newcommand{\Fbm}{\mathbf{F}^{-}}
\newcommand{\Gbm}{\mathbf{G}^{-}}
\newcommand{\Hbm}{\mathbf{H}^{-}}

\newcommand{\Px}[2]{\frac{\partial #1}{\partial #2}}
\newcommand{\Pxk}[3]{\frac{\partial^{#1} #2}{\partial #3^{#1}}}
\newcommand{\mbf}[1]{\mathbf{#1}}
\newcommand{\mbb}[1]{\boldsymbol{#1}}
\newcommand{\dr}[1]{\,\mathrm{d}{#1}}
\newcommand{\ull}[1]{\underline{\underline{#1}}}
\newcommand{\ul}[1]{\underline{#1}}

\newcommand{\Lh}{\hat L}
\newcommand{\Lxy}{{\bf L}(x)}
\newcommand{\Lxyp}{{\bf L}^{'}(x)}
\newcommand{\Lxymo}{{\bf L}(-1)}
\newcommand{\Lxypo}{{\bf L}(+1)}

\newcommand{\Lxe }{{\bf L}(\eta_l;x)}
\newcommand{\Lxep}{{\bf L}^{\prime}(\eta_l;x)}

\newcommand{\Yeig}{{\cal Y}}
\newcommand{\yV}{{\bf x}}
\newcommand{\etV}{{\boldsymbol \eta}}
\newcommand{\ez}{{\bf e_0}}

\newcommand{\yVa}{{\bf x}}
\newcommand{\zVa}{\Bar{\bf x}}
\newcommand{\zVF}{\Bar{\Bar{\bf x}}}

\newcommand{\fVya}{{\bf f}(\yVa)}
\newcommand{\fVza}{{\bf f}(\zVa)}
\newcommand{\dfVya}{{\bf f}'(\yVa)}
\newcommand{\dfVza}{{\bf f}'(\zVa)}

\newcommand{\Lb}{\bar L}
\newcommand{\Lxya}{{\bf L}(x;\yVa)}
\newcommand{\Lxza}{{\bf L}(x;\zVa)}

\newcommand{\IhlBa}{_m{\bar   {\cal I}}_n}
\newcommand{\IlhBa}{_n{\bar   {\cal I}}_m}
\newcommand{\IllBa}{_n{\bar   {\cal I}}_n}
\newcommand{\IhhBa}{_m{\bar   {\cal I}}_m}

\newcommand{\IhlBr}{_m{\breve {\cal I}}_n}
\newcommand{\IlhBr}{_n{\breve {\cal I}}_m}
\newcommand{\IllBr}{_n{\breve {\cal I}}_n}
\newcommand{\IhhBr}{_m{\breve {\cal I}}_m}

\newcommand{\DlBa}{{\bar   {\cal D}}_n}
\newcommand{\DlBr}{{\breve {\cal D}}_n}
\newcommand{\DhBa}{{\bar   {\cal D}}_m}
\newcommand{\DhBr}{{\breve {\cal D}}_m}

\newcommand{\Uh}{{\cal U}_m}
\newcommand{\Vh}{{\cal V}_m}
\newcommand{\Ah}{{\cal C}_m}
\newcommand{\Axh}{{\cal A_x}_m}

\newcommand{\ILtoH}{_m{\mathcal I}_n}
\newcommand{\IHtoL}{_n{\mathcal I}_m}
\newcommand{\evHO}{{\bar {\bf e}}(x_L)}
\newcommand{\evHM}{{\bar {\bf e}}(x_R)}

\newcommand{\Eh}{\bar {\mathcal E}}
\newcommand{\emat}{\mathcal{E}}

\newcommand{\fhb}{{\bar{\hat {\bf f}}}(\uu)}

\newcommand{\uL}{{u}(x_L)}
\newcommand{\uR}{{u}(x_R)}
\newcommand{\gL}{{g}(x_L)}
\newcommand{\gR}{{g}(x_R)}
\newcommand{\evO}{{\bf e}(x_L)}
\newcommand{\evM}{{\bf e}(x_R)}

\newcommand{\uu}{\bar {\bf u}}
\newcommand{\uh}{\bar {u}}
\newcommand{\ghL}{{\bar g}(x_L)}
\newcommand{\ghR}{{\bar g}(x_R)}
\newcommand{\uhL}{{\bar u}(x_L)}
\newcommand{\uhR}{{\bar u}(x_R)}
\newcommand{\UUmat}{\bar {\mathcal u}}
\newcommand{\DmatH}{\bar {\mathcal D}}
\newcommand{\PmatH}{\bar {\mathcal P}}
\newcommand{\PinvH}{{\bar {\mathcal P}}^{-1}}
\newcommand{\QmatH}{\bar {\mathcal Q}}
\newcommand{\BmatH}{\bar {\mathcal B}}
\newcommand{\DmatL}{{\mathcal D}}
\newcommand{\PmatL}{{\mathcal P}}
\newcommand{\QmatL}{{\mathcal Q}}
\newcommand{\BmatL}{{\mathcal B}}
\newcommand{\epvec}{\bf \epsilon}

\newcommand{\fVx}{{\bf f}(\xN)}
\newcommand{\dfVx}{{\bf f}'(\xN)}

\newcommand{\albar}{\bar \alpha}
\newcommand{\cii}{c_{11}}
\newcommand{\ciih}{\hat {c}_{11}}
\newcommand{\ciihm}{\hat {c}_{11}^{(-)}}
\newcommand{\ciihp}{\hat {c}_{11}^{(+)}}
\newcommand{\ciib}{\bar {c}_{11}}
\newcommand{\cijh}{\hat {c}_{ij}}
\newcommand{\cjih}{\hat {c}_{ji}}
\newcommand{\wim}{w^{(-)}_i}
\newcommand{\wip}{w^{(+)}_i}
\newcommand{\qim}{q^{(-)}_i}
\newcommand{\qip}{q^{(+)}_i}
\newcommand{\psim}{\psi^{(-)}_i}
\newcommand{\psip}{\psi^{(+)}_i}
\newcommand{\XiN}{\boldsymbol{\Theta}}
\newcommand{\Xim}{\Theta^{(-)}}
\newcommand{\Xip}{\Theta^{(+)}}

\newcommand{\fssc}{\mathbf{f}^{ssc}\left(q^{(-)}_i,q^{(+)}_i\right)}
\newcommand{\fssllf}{\mathbf{f}^{ssllf}\left(q^{(-)}_i,q^{(+)}_i\right)}

\newcommand{\fllf}{\mathbf{f}^{llf}(q^{(-)}_i,q^{(+)}_i)}
\newcommand{\fbm}{\mathbf{\overline{f}}^{(I)(-)}}
\newcommand{\fbp}{\mathbf{\overline{f}}^{(I)(+)}}
\newcommand{\fss}{\mathbf{f}^{ssr}\left(q^{(-)}_i,q^{(+)}_i\right)}
\newcommand{\fs}{\mathbf{f}^{sr}\left(q^{(-)}_i,q^{(+)}_i\right)}
\newcommand{\fssb}{\mathbf{f}^{ssr}\left(\mathbf{q}, \mathbf{g}_{E}\right)}
\newcommand{\half}{\frac{1}{2}}
\newcommand{\uave}{\bar u}
\newcommand{\psib}{\bar{\psi}}
\newcommand{\fbs}{\bar{f_{s}}}
\newcommand{\FMs}{\bar{F_{s}}}
\newcommand{\fsscb}{\mathbf{f}^{sr}}

\begin{frontmatter}



\title{Entropy stable wall boundary conditions for the three-dimensional compressible Navier--Stokes equations}


\author[NASA-LaRC-CASB]{Matteo Parsani\corref{cor1}}
\ead{matteo.parsani@nasa.gov}

\author[NASA-LaRC-CASB]{Mark H. Carpenter}
\ead{mark.h.carpenter@nasa.gov}

\author[NASA-LaRC-CASB]{Eric J. Nielsen}
\ead{eric.j.nielsen@nasa.gov}

\cortext[cor1]{Corresponding author.}
\address[NASA-LaRC-CASB]{Computational AeroSciences Branch, NASA Langley Research 
  Center (LaRC), Hampton, VA 23681, USA}

\begin{abstract}
Non-linear entropy stability and a summation-by-parts framework are 
used to derive entropy stable wall boundary conditions 
for the three-dimensional compressible Navier--Stokes equations. A semi-discrete entropy estimate for the entire 
domain is achieved when the new boundary conditions are coupled with an
entropy stable discrete interior operator.
The data at the boundary are weakly imposed using a penalty
flux approach and a simultaneous-approximation-term penalty technique.
Although discontinuous spectral collocation operators on unstructured grids are used herein for the 
purpose of demonstrating their robustness and efficacy, the new boundary conditions are 
compatible with any diagonal norm summation-by-parts spatial operator, including 
finite element, finite difference, finite volume, discontinuous Galerkin, and flux reconstruction/correction procedure via reconstruction schemes. 
The proposed boundary treatment is tested for three-dimensional subsonic
and supersonic flows. 
The numerical computations corroborate
the non-linear stability (entropy stability) and accuracy of the boundary conditions.
\end{abstract}


\end{frontmatter}


\section{Introduction}\label{sec:intro}

During the last twenty years, scientific computation has become 
a broadly-used technology in all fields of science and engineering due to 
a million-fold increase in 
computational power and the development of advanced 
algorithms. However, the great frontier is in the challenge posed by high-fidelity
simulations of real-world systems, that is, in truly transforming computational science into a fully
predictive science. Much of scientific computation's potential remains unexploited--in areas such as engineering design, energy assurance, material science, Earth science,
medicine, biology, security and fundamental science--because the scientific challenges are far too gigantic and complex
for the current state-of-the-art computational resources \cite{mathematical-exascale}.

In the near future, the transition from petascale to exascale systems will
provide an unprecedented opportunity to attack these global challenges using
modeling and simulation. However, exascale programming models will
require a revolutionary approach, rather than the incremental approach of
previous projects. Rapidly
changing high performance computing (HPC) architectures, which often include
multiple levels of parallelism through heterogeneous architectures, will require new paradigms 
to exploit their full potential. However, the 
complexity and diversity of issues in most of the science community are such 
that increases in computational power alone will not be enough to reach the required
goals, and new algorithms, solvers and physical models with better 
mathematical and numerical properties must continue to be developed and 
integrated into new generation supercomputer systems.

In computational fluid dynamics (CFD), next generation numerical algorithms for use in large 
eddy simulations (LES) and hybrid Reynolds-averaged Navier--Stokes (RANS)-LES 
simulations will undoubtedly rely on efficient high-order accurate formulations
(see, for instance,
\cite{FLD:FLD3767,les-dg-persson-aiaa,2010:IHS:1805352.1805674,Hartmann20114268,yano-aiaa,
burgess-mavriplis-iccfd7-2012,FLD:FLD3740,Blaise2013416,les-muffler,dg-rans-icosahom12-bassi,
SH14a,Wang201413,FLD:FLD3943,gmdd-7-4251-2014}).
Although high order techniques are well suited for smooth solutions, numerical 
instabilities may occur if the flow contains discontinuities or under-resolved 
physical features. A variety of mathematical stabilization strategies
are commonly used to cope with this problem (e.g., filtering 
\cite{Hesthaven:2007:NDG:1557392}, 
weighted essentially non-oscillatory (WENO) \cite{Zhu2013200}, artificial 
dissipation, and limiters\cite{FLD:FLD3767}),
but their use for practical complex flow applications in realistic geometries is 
still problematic.

A very promising and mathematically rigorous alternative is to focus directly 
on discrete operators that are
non-linearly stable (entropy stable) for the Navier--Stokes equations. These 
operators simultaneously conserve mass, momentum, energy, and 
enforce a secondary entropy constraint.  
This strategy begins by identifying a non-linear 
neutrally stable flux for the Euler equations. 
An appropriate amount of dissipation can then be added to achieve the 
desired smoothness of the solution. Regardless of whether dissipation is added, 
enforcing a semi-discrete entropy constraint enhances the stability of 
the base operator.

The idea of enforcing entropy stability in numerical operators is quite old 
\cite{hughes1986}, and is commonly used for low-order operators 
\cite{Tadmor2003,Tadmor1987}. An extension of these techniques to include
high-order accurate operators recently appears in references 
\cite{FisherCarpenter2013NASA,FisherCarpenter2013JCPb,CarpenterFisher2013AIAA,svard-zamp-2012}
and provides a general procedure for developing entropy conservative and entropy
stable, diagonal norm summation-by-parts (SBP) operators for the compressible
Navier--Stokes equations. The strong conservation form representation allows 
them to be readily extended to capture shocks via a comparison approach 
\cite{Tadmor2003,FisherCarpenter2013JCPb}. The generalization to multi-domain 
operators follows immediately using simultaneous-approximation-term (SAT) 
penalty type interface conditions \cite{cng2009}, whereas the extension to 
three-dimensional (3D) curvilinear coordinates is obtained by using an
appropriate coordinate transformation which satisfies the discrete 
geometric conservation law \cite{Fisher2012dissertation}.  See reference 
LeFloch and Rohde \cite{doi:10.1137/S0036142998345256} for a more comprehensive 
discussion on high order schemes and entropy inequalities.
Therein, the focus is on the approximation of under-compressive, 
regularization-sensitive, non-classical solutions of hyperbolic systems of conservation laws 
by high-order accurate, conservative, and semi-discrete finite difference
methods.

Several major hurdles 
remain, however, on the path towards complete entropy stability of the 
compressible Navier--Stokes equations including shocks. A major obstacle is 
the need for solid wall viscous boundary conditions that 
preserve the entropy stability property of the interior operator. In fact, 
practical experience indicates that numerical instability frequently originates 
at solid walls, and the interaction of shocks with these 
physical boundaries is particularly challenging for high order formulations. An 
important step towards entropy stable boundary conditions appears in the work of 
Sv\"{a}rd and {\"O}zcan \cite{svard-no-penetration-bc-JSC}. Therein, entropy 
stable boundary conditions for the compressible Euler equations are reported for the 
far-field and for the Euler no-penetration wall conditions, in the context of 
finite difference approximations. 

The focus herein is on building non-linearly stable wall boundary conditions for the 
compressible Navier--Stokes equations; primarily a task of developing stable wall boundary 
conditions for the viscous terms. At the semi-discrete level, the proposed 
boundary treatment mimics exactly the boundary contribution obtained by applying the entropy 
stability analysis to the continuous, compressible Navier--Stokes equations. Furthermore, 
the new technique enforces the Euler no-penetration wall condition as well as 
the remaining no-slip and thermal wall conditions in a weak sense. The thermal boundary condition is 
imposed by prescribing the heat entropy flow (or heat entropy transfer), which is 
the primary means for exchanging entropy between two thermodynamic 
systems connected by a solid viscous wall.
Note that the entropy
flow at a wall is a quantity that in some experiments is directly or
indirectly available (e.g., through measurements 
of the wall heat flux and temperature in some 
supersonic or hypersonic wind tunnel experiments), or can be estimated from geometrical
parameters and fluid flow conditions for the problem at hand. For
fluid-structure interaction simulations, 
(e.g., supersonic and hypersonic flow past aerospace vehicles, heat-exchangers), 
the entropy flow can be numerically computed at no additional cost
while numerically solving the coupled systems of partial differential equations of the 
continuum mechanics and fluid dynamics. 

Historically, most boundary condition analysis for the 
Navier--Stokes equations is performed at the linear level by linearizing about a known state;
a rich collection of literature is available
\cite{gustafsson-bc-book-1995,nordstrom-well-posedness-sinum-2005,svard-no-slip-jcp-2008,berg2011}.
The non-linear wall boundary conditions developed herein are fundamentally different 
from those derived using linear analysis and cannot rely on a 
complete mathematical theory. In fact, a fundamental shortcoming that 
limits further development of any entropy stable boundary conditions is the incomplete development of the 
analysis at the continuous level for proving well-posedness of the compressible 
Navier--Stokes equations.  
Nevertheless, the boundary conditions proposed herein is extremely powerful because they
provide a mechanism for ensuring the non-linear stability in the $L^2$ norm of the 
semi-discretized compressible Navier--Stokes equations. In fact, they
allow for a priori bounds on the entropy function when imposing ``solid viscous wall'' boundary conditions.
The new boundary conditions are easy to implement and 
compatible with any diagonal norm SBP spatial operator, including 
finite element (FE), finite volume (FV), finite difference (FD) schemes
and the more recent class of high-order accurate methods which include the 
large family of discontinuous Galerkin (DG) 
discretizations \cite{dg-survey-shu-2013} and flux reconstruction (FR) schemes \cite{FR-jameson-2011}.

The robustness and accuracy of the complete semi-discrete operator 
(i.e., the entropy stable interior operator coupled with the new boundary 
conditions) is demonstrated by computing subsonic and supersonic flows past a
3D square cylinder
without any stabilization technique (e.g., artificial dissipation, 
filtering, limiters, de-aliasing, etc.), a feat
unattainable with several alternative approaches to wall boundary conditions 
based on linear analysis. 
In fact, instabilities are observed when wall boundary conditions designed with
linear analysis are used in combination with highly 
clustered grids and/or high order polynomials, or very coarse grids near
solid walls, which yield unresolved physical flow features.

The paper is organized as follows. In Section \ref{sec:cnse}, the 
compressible Navier--Stokes equations, their entropy function and symmetrized 
form are introduced. Section \ref{sec:well-posedness} presents the
entropy analysis of the viscous wall boundary conditions at the 
continuous level. Section \ref{sec:inviscid-flux-condition} provides a
discussion of the inviscid flux condition which allows the construction of high-order
accurate entropy conservative and entropy stable fluxes at the semi-discrete
level, for the interior operator. The new entropy stable wall boundary conditions and their non-linear
stability  (entropy stability) analysis are presented in Section \ref{sec:ss-no-slip-bc}. Finally, 
the accuracy and high level robustness of the proposed approach in combination with a 
discontinuous high-order accurate entropy stable interior operator are demonstrated
in Section \ref{sec:validation}. Conclusions are discussed in Section 
\ref{sec:conclusions}.


\section{The compressible Navier--Stokes equations}\label{sec:cnse}
Consider a fluid in a domain $\Omega$ with boundary surface denoted by 
$\partial\Omega$, without radiation and external volume forces. In this context,
the compressible Navier--Stokes equations, equipped with suitable boundary and
initial conditions, may be expressed in the form 
\begin{equation}
\label{eq:cnse} 
\begin{aligned}
  & \frac{\partial q}{\partial t} + \frac{\partial f^{(I)}_i}{\partial x_i} =
  \frac{\partial f^{(V)}_i}{\partial x_i}, \quad x \in \Omega, \quad t \in 
 [0,\infty), \\
   & q|_{\partial\Omega} = g^{(B)}(x,t),  \quad x \in 
   \partial\Omega, \quad t \in 
 [0,\infty),   \\
   & q(x,0) = g^{(0)}(x), \quad x \in \Omega,
\end{aligned}
\end{equation}
where the Cartesian coordinates, $x = \left(x_1,x_2,x_3\right)^\top$, and time, 
$t$, are the
independent variables. Note that in \eqref{eq:cnse} Einstein notation is used. 
The vectors $q(x,t)$, $f^{(I)}_i = f^{(I)}_i(q)$, and 
$f^{(V)}_i= f^{(V)}_i(q,\nabla q)$ are the 
conserved variables, and the inviscid and viscous fluxes in the $i$ direction, 
respectively.\footnote{The symbol $\nabla q$ denotes the gradient of the 
conservative variables.} Both boundary data, $g^{(B)}$, and initial data,
$g^{(0)}$,
are assumed to be bounded, $L^2 \cap L^{\infty}$. Furthermore, $g^{(B)}$ is
also assumed to contain (linearly) well-posed Dirichlet and/or Neumann and/or Robin data. 
The vector of the conservative variables is
\begin{equation} \label{eq:cons-vec}
  q = \left(\rho,\rho u_1,\rho u_2,\rho u_3,\rho E\right)^\top,
\end{equation}
where $\rho$ denotes the density, $u=\left(u_1,u_2,u_3\right)^\top$ is the 
velocity vector, and $E$ is the specific total energy, which is the sum of the
specific internal energy, $e$, (defined later) and 
the kinetic energy, $\frac{1}{2}u_j u_j$. The convective fluxes are
\begin{equation}\label{eq:inv-flux}
  f^{(I)}_i = \left(\rho u_i, \rho u_i u_1 + \delta_{i1}p,\rho u_i u_2 + \delta_{i2}p,
  \rho u_i u_3 + \delta_{i3}p, \rho u_i H \right)^\top,
\end{equation}
where $p$, $H$ and $\delta_{ij}$ are the pressure, the specific total enthalpy
and the Kronecker delta, respectively. The viscous fluxes are
\begin{equation}\label{eq:visc-flux}
  f^{(V)}_i = \left(0, \tau_{i1}, \tau_{i2}, \tau_{i3},
  \tau_{ji}u_j-\mathtt{q}_i\right)^\top,
\end{equation}
where the shear stress tensor, $\tau_{ij}$, assuming a zero bulk viscosity, is defined as
\begin{equation}\label{eq:shear-stress}
  \tau_{ij} = \mu \left(\frac{\partial u_i}{\partial x_j} + \frac{\partial
    u_j}{\partial x_i} -
  \delta_{ij} \, \frac{2}{3}\frac{\partial u_k}{\partial x_k}\right),
\end{equation}
and the heat flux, according to the Fourier heat conduction law, is 
\begin{equation}\label{eq:heat-flux}
  \mathtt{q}_i = -\kappa \, \frac{\partial T}{\partial x_i}.
\end{equation}
The symbols $T$, $\mu=\mu\left(T\right)$ and $\kappa=\kappa\left(T\right)$  
which appear in \eqref{eq:shear-stress} and \eqref{eq:heat-flux} denote the 
temperature, the dynamic viscosity and the thermal conductivity, respectively.
The viscous fluxes in \eqref{eq:visc-flux} can also be
expressed as 
\begin{equation}\label{eq:visc-flux-cij}
  f^{(V)}_i = c_{ij} \, \frac{\partial q}{\partial x_j} = c_{ij}^{\prime} 
  \frac{\partial v}{\partial x_j},
\end{equation}
where $c_{ij}$ and $c_{ij}^{\prime}$ are variable five-by-five
matrices\footnote{The coefficient matrices $c_{ij}$ and $c_{ij}^{\prime}$
depends on the solution variables.} and
$v=\left(\rho,u_1,u_2,u_3,T\right)^{\top}$ is the vector of the primitive 
variables. The functional form of the 
matrices $c_{ij}^{\prime}$ is given in \ref{app:cij}. As we will show
later, relation \eqref{eq:visc-flux-cij} is 
a very convenient form for the entropy stability analysis. 

The constitutive relations for a calorically perfect gas are
\begin{equation} \label{eq:specific-total-energy}
  e = c_v \, T, \quad  h = H-\frac{1}{2}u_j \, u_j = c_p \, T,
\end{equation}
where the symbols $c_v$ and $c_p$ denote the specific heat capacity 
at constant volume and constant pressure, respectively, and 
\begin{equation}
 p = \rho R \, T, \quad R = \frac{R_u}{MW},
\end{equation}
where  $R_u$ is the universal gas constant and $MW$ is the molecular weight of
the gas. The speed of sound for a perfect gas is
\begin{equation}
 c = \sqrt{\gamma R \, T}, \quad \gamma = \frac{c_p}{c_p-R} .
\end{equation}
In the entropy analysis that will follow, the definition of the thermodynamic 
entropy is the explicit form,
\begin{equation}\label{eq:math-entropy}
  s = \frac{R}{\gamma-1}\log\left(\frac{T}{T_{\infty}}\right) 
  - R \log\left(\frac{\rho}{\rho_{\infty}}\right),
\end{equation}
where $T_{\infty}$ and $\rho_{\infty}$ are the reference temperature and density, 
respectively.


Note that in practical situations, most simulations are performed 
in computational space, that is, by transforming all grid elements in physical 
space to standard elements in computational space 
via a smooth mapping. However, to keep the notation as simple as possible,
a uniform Cartesian 
grid is considered in the derivations presented herein. The extension to generalized 
curvilinear coordinates and unstructured grids follows
immediately if the transformation from computational to physical space
preserves the semi-discrete geometric conservation (GCL) \cite{Fisher2012dissertation}.


\subsection{Entropy function and entropy variables of the compressible 
  Navier--Stokes equations} \label{sec:continuousentropyanalysis}

In the linear hyperbolic framework, $L^2$ stability is sought as a discrete 
analogue for a priori energy estimates available in the differential formulation
(e.g., Richtmyer and Morton \cite{richtmyer-ivp-book-1967} and
Gustafsson, Kreiss and Oliger \cite{gustafsson-bc-book-1995}; 
Kreiss and Lorenz \cite{kreiss-stability-actan-1998}). In the context of
non-linear problems dominated by non-linear convection, we seek entropy stability
as a discrete analogue for the corresponding statement in the differential 
formulation \cite{Tadmor2003}. Moreover, for systems of conservation laws, stability with respect to a 
mathematical 
entropy function is considered as an admissibility criterion for selecting 
physically relevant weak solutions. In fact, the entropy condition plays a 
decisive role in the theory and numerics of such problems as shown, for
instance, 
by Lax \cite{lax-siam-conference-1972} and Smoller 
\cite{smoller-entropy-book-1994}.

Harten \cite{harten1983} and Tadmor \cite{tadmor1984} showed that systems of conservation 
laws are symmetrizable if and only 
if they are equipped with a convex mathematical entropy function. Given a set of conservation 
variables $q(x,t)$, the entropy variables which symmetrize the system are 
defined as the derivatives of the mathematical entropy function with
respect to $q(x,t)$. Hughes and co-authors \cite{hughes1986} extended these 
ideas to the compressible Navier--Stokes equations \eqref{eq:cnse}. Therein, it 
is shown that
the mathematical entropy must be an affine function of the physical (or thermodynamic) 
entropy 
function and that semi-discrete solutions obtained from a weighted residual 
formulation based on entropy variables will respect the second law of
thermodynamics. Hence, it is again found that the entropy function and the entropy variables are critical 
ingredients in the design of numerical schemes exhibiting non-linear stability.

\begin{definition}
\label{def:entropyfunction}
A scalar function $S=S(q)$ is an entropy function of Equation 
\eqref{eq:cnse} if it satisfies the following conditions:
\begin{itemize}
\item 
The function $S(q)$ when differentiated with respect to the conservative
variables (i.e., $\partial S/\partial q$) 
simultaneously contracts all the inviscid spatial fluxes as follows
\begin{equation}\label{eq:entropymagic}
\begin{aligned}
  \frac{\partial S}{\partial q} \, \frac{\partial f^{(I)}_i}{\partial x_i} \:=\: \frac{\partial S}{\partial q} \frac{\partial
  f^{(I)}_i}{\partial q} \, \frac{\partial q}{\partial x_i} 
  \:=\: \frac{\partial F_i}{\partial q} \, \frac{\partial q}{\partial x_i} \:=\:
  \frac{\partial F_i}{\partial x_i}, \quad  i = 1,2,3.
\end{aligned}
\end{equation}
The components of the contracting vector, $\partial S/\partial q$, are the entropy
variables denoted as $w^{\top}\:=\:\partial S/\partial q$. $F_i(q), \, i=1,2,3,$ are the 
entropy fluxes in the three Cartesian directions.
\item
The new entropy variables, $w$, symmetrize Equation \eqref{eq:cnse}: 
\begin{equation}
\label{eq:entropyChainRule}
\begin{aligned}
  \frac{\partial q}{\partial t} + \frac{\partial f^{(I)}_i}{\partial x_i} - 
  \frac{\partial f^{(V)}_i}{\partial x_i}  
  \:=\: \frac{\partial q}{\partial w} \, \frac{\partial w}{\partial t} + 
  \frac{\partial f^{(I)}_i}{\partial w} \, \frac{\partial w}{\partial x_i}  
  - \frac{\partial}{\partial x_i}\left(\widehat{c}_{ij} \, \frac{\partial w}{\partial x_j}
  \right)  = 0,  \quad i=1,2,3
\end{aligned}
\end{equation}
with the symmetry conditions:  
$\partial q/\partial w = \left(\partial q/ \partial w\right)^{\top}$, 
$\partial f^{(I)}_i/\partial w = \left(\partial f^{(I)}_i/\partial w\right)^{\top}$  and 
$\widehat{c}_{ij} = {\widehat{c}_{ij}}^{\top}$,
with the matrices $\,\widehat{c}_{ij}$ included in \ref{app:cij}.
\end{itemize}
The mathematical entropy is convex, meaning that the Hessian, 
$\partial^2 S/\partial q^2\:=\:\partial w/\partial q$, is symmetric positive 
definite (SPD),
\begin{equation}
  \quad \zeta^T \, \frac{\partial^2 S}{\partial q^2} \, \zeta >0, \quad \forall \zeta \neq 0,
\end{equation}
and yields a one-to-one mapping from conservation variables, $q$, to
entropy variables, $w^{\top}\:=\:\partial S/\partial q$. Likewise, 
$\partial w/ \partial q$ is SPD because 
$\partial q /\partial w = \left({\partial w/\partial q}\right)^{-1}$ and SPD matrices are 
invertible. The entropy and corresponding entropy flux are often denoted 
as entropy--entropy flux pair, $(S,F)$, \cite{Tadmor2003}. If the entropy 
function $S(q)$ is convex, a bound on its global
estimate can be converted into an a priori estimate on the solution of \eqref{eq:cnse}
in suitable $L^p$ space \cite{dafermos-book-2010}. 
\end{definition}

The symmetry of the matrices $\partial q /\partial w$, 
$\partial f^{(I)}_i/\partial w$ 
indicates that the conservative variables, $q$,
and the inviscid fluxes, $f^{(I)}_i$, are Jacobians of scalar functions with respect to the 
entropy variables,
\begin{equation}
  q^{\top} = \frac{\partial \varphi}{\partial w}, \quad 
  \left(f^{(I)}_i\right)^{\top} = \frac{\partial \psi_i}{\partial w},
\end{equation}
where the non-linear function $\varphi$ is called the potential
and $\psi_i, \, i=1,2,3$, are called the potential fluxes \cite{Tadmor2003}. 
Just as the entropy function is convex with respect to
the conservative variables ($\partial^2 S/\partial q^2$ is SPD), 
the potential function is convex with respect to the entropy variables.

Godunov \cite{Godunov1961} proved that (see reference
\cite{harten1983} for a detailed summary of the proof):
\begin{theorem}
If Equation \eqref{eq:cnse} can be symmetrized by introducing new variables 
$w$, and $q$ is a convex function of $\varphi$ (the so-called potential), then 
an entropy function $S(q)$ is given by 
\begin{equation}
\label{eq:GodSpotential}
\varphi = w^{\top} q - S, 
\end{equation}
and the entropy fluxes satisfy
\begin{equation}
\label{eq:GodFpotential}
\psi_i = w^{\top} f^{(I)}_i - F_i, 
\end{equation}
where $\psi_i, i=1,2,3,$, are the so-called potential fluxes. The potential 
and the corresponding potential flux are usually denoted as potential--potential 
flux pair, $(\varphi,\psi)$, \cite{Tadmor2003}.
\end{theorem}

In the specific case of the compressible Navier--Stokes equations, the 
entropy--entropy flux pair is
\begin{equation}
\label{eq:nseentropy}
 S = -\rho \, s, \quad F_i = -\rho \, u_i \, s, \quad i =1,2,3,
\end{equation}
and the potential--potential flux pair is
\begin{equation}
 \varphi = \rho \, R, \quad \psi_i = \rho \, u_i \, R, \quad i=1,2,3.
\end{equation}
Note that the mathematical entropy has the opposite sign
of the thermodynamic entropy. To avoid confusion, herein entropy
refers to the mathematical entropy unless otherwise noted. 
The entropy variables using the pair in \eqref{eq:nseentropy} result in
\begin{equation}\label{eq:entropy-variables-cnse}
  w = \left(\frac{\partial S}{\partial q}\right)^{\top} = \left( \frac{h}{T} - s - \frac{u_i u_i}{2T}, \frac{u_1}{T},
  \frac{u_2}{T},\frac{u_3}{T},-\frac{1}{T}\right)^{\top}.
\end{equation}
A sufficient condition to ensure the convexity of the entropy function $S(q)$
(and, hence, a one-to-one mapping between the entropy variables
\eqref{eq:entropy-variables-cnse} and
the conservative variables  \eqref{eq:cons-vec}) is that
$\rho, T>0$ (for the proof see, for instance, Appendix B.1
in \cite{FisherCarpenter2013JCPb}). Expressly:
\[ \zeta^T \, \frac{\partial^2 S}{\partial q^2} \, \zeta^T > 0, \quad \forall \zeta \neq 0, \quad \rho,T > 0. \]
This restriction on density and temperature weakens the entropy proof, making 
it less than full measure of non-linear stability. Another
mechanism must be employed to bound $\rho$ and $T$ away from zero to guarantee 
stability and positivity;  positivity preservation will not be considered herein.
 
\begin{remark}
The semi-discrete entropy stability
does not necessarily lead to fully discrete stability as is usually the case
for linear partial differential equations (PDEs). However, as noted by Tadmor
\cite{Tadmor2003}, entropy stability is enhanced by fully implicit time discretization.
For example, the fully implicit backward Euler time discretization is unconditionally
entropy-stable and is responsible for additional entropy dissipation.
In contrast, explicit time discretization, leads to
entropy production. Thus, the entropy stability of explicit schemes hinges
on a delicate balance between temporal entropy production and spatial entropy
dissipation. For example, the fully explicit Euler time discretization does not
conserve entropy except in the case of linear fluxes
\cite{doi:10.1137/S0036142998345256}. Consequently, both the
fully explicit and fully implicit Euler differencing do not respect (non-linear) 
entropy conservation, independent of the spatial discretization. Fully
discrete entropy conservation is offered, for instance, by Crank--Nicolson time differencing \cite{Tadmor2003}.
\end{remark}


\section{No-slip boundary conditions: Continuous analysis} \label{sec:well-posedness}

The problem of well-posed boundary conditions is an essential question in many
area of physics. For the two- (2D) and three-dimensional (3D) Navier--Stokes equations,
the number of boundary conditions implying well-posedness can be obtained using
the Laplace transform technique \cite{gustafsson-bc-book-1995}; a complicated
procedure for system of PDEs like the compressible Navier--Stokes equations. Nordstr{\"o}m and Sv{\"a}rd 
\cite{nordstrom-well-posedness-sinum-2005} proposed an alternative semi-discrete approach
to arrive at the number and type of boundary
conditions for a general time-dependent system of PDEs. This analysis was 
applied to the 3D compressible Navier--Stokes equations for different flow regimes
and the case of the Euler no-penetration velocity condition. 

In 2008, Sv{\"a}rd and Nordstr{\"o}m 
\cite{svard-no-slip-jcp-2008} showed that the
no-slip boundary conditions together with a boundary condition on the
temperature imply stability for the linearized compressible Navier--Stokes equations. 
Their result, can also be generalized to assess the stability of 
the non-linear problem for smooth solutions as indicated in 
\cite{kreiss-bc-book-1989,gustafsson-bc-book-1995} and references therein. In
2011, Berg
and Nordstr{\"o}m \cite{berg2011} proved that the
no-slip boundary conditions together with a thermal Robin boundary condition
also imply stability for the same linearized equations.
In this section we address the non-linear stability (entropy stability) of the wall boundary conditions for 
the (non-linear) compressible Navier--Stokes equations given in
\eqref{eq:cnse}. The aim is to derive a sharp a priori bound for the entropy function which
holds for smooth and non-smooth solutions. 

Contracting the system of equations \eqref{eq:cnse} with the 
entropy variables and using the relations given in \eqref{eq:entropymagic} and 
\eqref{eq:entropyChainRule} results in the differential form of the (scalar) entropy 
equation:
\begin{equation}\label{eq:entropyequation1}
\begin{aligned}
  \frac{\partial S}{\partial q} \frac{\partial q}{\partial t} + 
  \frac{\partial S}{\partial q} \frac{\partial f^{(I)}_i}{\partial x_i}  = 
  \frac{\partial S}{\partial t} + \frac{\partial F_i}{\partial x_i} 
  =
  \frac{\partial S}{\partial q} \frac{\partial f^{(V)}_i}{\partial x_i} 
  & = \frac{\partial}{\partial x_i}\left(w^{\top} f^{(V)}_i\right) - 
  \left(\frac{\partial w}{\partial x_i}\right)^{\top} f^{(V)}_i \\
  & = \frac{\partial}{\partial x_i}\left(w^{\top} f^{(V)}_i\right) -
  \left(\frac{\partial w}{\partial x_i}\right)^{\top} \widehat{c}_{ij} \, 
  \frac{\partial w}{\partial x_j}.
\end{aligned}
\end{equation}
Integrating Equation \eqref{eq:entropyequation1} over the domain yields a global 
conservation statement for the entropy in the domain,
\begin{equation}
 \label{eq:continuousentropyestimate-integration}
 \ddt \int_{\Omega} S \dr{{\mathbf{x}}} 
\:=\: \left[w^{\top} f^{(V)}_i - F_i \right]_{\partial \Omega}
 - \int_{\Omega} \left(\frac{\partial w}{\partial x_i}\right)^{\top} \widehat{c}_{ij} \, \frac{\partial w}{\partial x_j} \, \dr{{\bf x}}
\:=\: \left[w^{\top} f^{(V)}_i - F_i \right]_{\partial \Omega} \:-\: DT,
\end{equation}
where $DT$ is the viscous dissipation term,
\begin{equation}\label{eq:continuous-dissipation}
DT =\bigintsss_{\Omega} \left(\frac{\partial w}{\partial x_i}\right)^{\top} \widehat{c}_{ij} \, \frac{\partial w}{\partial x_j} \, \dr{{\bf x}}
= \bigintsss_{\Omega} 
\begin{pmatrix}
  \partial w/\partial x_1 \\ \partial w/\partial x_2  \\ \partial w/\partial x_3
\end{pmatrix}^{\top}
\begin{pmatrix}
  \widehat{c}_{11} & \widehat{c}_{12} & \widehat{c}_{13} \\
  \widehat{c}_{21} & \widehat{c}_{22} & \widehat{c}_{23} \\
  \widehat{c}_{31} & \widehat{c}_{32} & \widehat{c}_{33} \\
\end{pmatrix}
\begin{pmatrix}
  \partial w/\partial x_1 \\ \partial w/\partial x_2  \\ \partial w/\partial x_3
\end{pmatrix} \dr{{\bf x}}.
\end{equation}
This equation indicates that the entropy 
can only increase in the domain based on
data that convects and diffuses through the boundaries.
It is shown in \cite{hughes1986,FisherCarpenter2013JCPb} that the larger $15 \times 15$ 
coefficient matrix 
in \eqref{eq:continuous-dissipation} 
 is positive semi-definite, which makes the term $-DT$ in \eqref{eq:continuousentropyestimate-integration}
strictly dissipative. Therefore, the sign of the entropy 
change due to viscous dissipation is always negative.\footnote{From a physical
point of view, this means that the viscous dissipation always yields an
increase of the thermodynamic entropy.}


To simplify, we let the domain of interest, $\Omega$, be the unit cube 
$0\leq x_1,x_2,x_3 \leq 1$. Expanding the Einstein notation in equation 
\eqref{eq:continuousentropyestimate-integration} yields  
\begin{equation}\label{eq:continuousentropyestimate}
\begin{aligned}
  & \ddt \int_{\Omega} S \, dx_1 \, dx_2 \, dx_3 =  -DT  \\ 
  & \:+\: \int_{x_1=0}  \left[+F_1 - \, w^{\top}  
  \left(\widehat{c}_{11} \frac{\partial w}{\partial x_1} 
  +     \widehat{c}_{12} \frac{\partial w}{\partial x_2} +    
        \widehat{c}_{13} \frac{\partial w}{\partial x_3} \right)
                                                         \right] dx_2 \, dx_3 \\
  & \:+\: \int_{x_1=1}  \left[-F_1 +  \, w^{\top} 
  \left(\widehat{c}_{11} \frac{\partial w}{\partial x_1} 
  +     \widehat{c}_{12} \frac{\partial w}{\partial x_2} +    
        \widehat{c}_{13} \frac{\partial w}{\partial x_3} \right)
                                                         \right] dx_2 \, dx_3 \\
  & \:+\: \int_{x_2=0}  \left[+F_2 - \, w^{\top} 
  \left(\widehat{c}_{21} \frac{\partial w}{\partial x_1} 
  +     \widehat{c}_{22} \frac{\partial w}{\partial x_2} +    
        \widehat{c}_{23} \frac{\partial w}{\partial x_3} \right)
                                                         \right] dx_1 \, dx_3 \\
  & \:+\: \int_{x_2=1} \left[-F_2 + \, w^{\top} 
  \left(\widehat{c}_{21} \frac{\partial w}{\partial x_1} 
  +     \widehat{c}_{22} \frac{\partial w}{\partial x_2} +    
        \widehat{c}_{23} \frac{\partial w}{\partial x_3} \right)
                                                         \right] dx_1 \, dx_3 \\
  & \:+\: \int_{x_3=0}  \left[+F_3 - \, w^{\top} 
  \left(\widehat{c}_{31} \frac{\partial w}{\partial x_1} 
  +     \widehat{c}_{32} \frac{\partial w}{\partial x_2} +    
        \widehat{c}_{33} \frac{\partial w}{\partial x_3} \right)
                                                         \right] dx_1 \, dx_2 \\
  & \:+\: \int_{x_3=1} \left[-F_3 + \, w^{\top} 
  \left(\widehat{c}_{31} \frac{\partial w}{\partial x_1} 
  +     \widehat{c}_{32} \frac{\partial w}{\partial x_2} +    
        \widehat{c}_{33} \frac{\partial w}{\partial x_3} \right)
                                                         \right] dx_1 \, dx_2 \, .
\end{aligned}
\end{equation}
%
Consider the case of a wall placed at $x_1=0$, and assume that all the other boundaries terms are
entropy stable;  their contributions are neglected without loosing generality. 
Therefore, estimate \eqref{eq:continuousentropyestimate} reduces to
\begin{equation}\label{eq:continuousentropyestimate-1wall}
\begin{aligned}
  & \ddt \int_{\Omega} S \, dx_1 \, dx_2 \, dx_3 =  - DT  \\ 
  & \:+\: \int_{x_1=0}  \left[F_1 - \, w^{\top}  
  \left(\widehat{c}_{11} \frac{\partial w}{\partial x_1} 
  +     \widehat{c}_{12} \frac{\partial w}{\partial x_2} +    
        \widehat{c}_{13} \frac{\partial w}{\partial x_3} \right)
                                                         \right] dx_2 \, dx_3 \,.
\end{aligned}
\end{equation}
To bound the time derivative of the entropy, the right-hand-side (RHS) of Equation
\eqref{eq:continuousentropyestimate-1wall} requires boundary data. 
For a solid viscous wall, assuming linear analysis, 
four independent boundary conditions must be imposed 
 \cite{kreiss-bc-book-1989}.\footnote{Using the linear analysis, it can be shown
 that a
 solid viscous wall 
 behaves like a subsonic outflow 
\cite{nordstrom-well-posedness-sinum-2005}.} Three boundary conditions are the so-called 
no-slip boundary conditions, $u_1 = u_2 = u_3 = 0$ (i.e., the velocity vector
relative to the wall is the zero vector). In the linear settings (see, for instance, 
\cite{berg2011,svard-no-slip-jcp-2008}), the fourth condition is the 
gradient of the temperature normal to the wall, 
$(\partial T / \partial n)_{wall}$, (Neumann boundary condition; e.g., the adiabatic wall), 
or the temperature of the wall, $T_{wall}$, (the Dirichlet or isothermal wall boundary condition), 
or a mixture of Dirichlet and Neumann conditions (the Robin boundary condition). 
These four boundary conditions lead to energy stability 
(linear stability); see, for instance, \cite{svard-no-slip-jcp-2008,berg2011}. 
In the remainder of this section, we will show the type and the form of the wall boundary
conditions that have to be imposed to bound estimate \eqref{eq:continuousentropyestimate-1wall} and, hence, to attain entropy stability.

Consider the inviscid contribution to the boundary terms in 
\eqref{eq:continuousentropyestimate-1wall}. 
\begin{theorem}\label{th:continous-inviscid-bound}
The no-slip boundary conditions $u_1 = u_2 = u_3 = 0$ bound the inviscid
contribution to the time derivative of the entropy in equation
\eqref{eq:continuousentropyestimate-1wall}.
\end{theorem}
\begin{proof}
Equation \eqref{eq:GodFpotential} provides the definition of the entropy flux, $F_1$:
\begin{equation}
  F_1 = w^{\top} f^{(I)}_1 - \psi_1, \quad \psi_1 = \rho u_1 R.
\end{equation}
Substituting the no-slip conditions, $u_1 = u_2 = u_3 = 0$, into the definition 
of the inviscid flux, $f^{(I)}_1$, (Equation \ref{eq:inv-flux})
and the condition $u_1 = 0$ into the definition 
of $\psi_1$, yields the desired result  $F_1 = 0$.
\end{proof}

Consider now the viscous contribution to the boundary terms in 
\eqref{eq:continuousentropyestimate-1wall}. 
\begin{theorem}\label{th:continous-viscous-bound}
  The boundary condition, 
  \begin{equation}\label{eq:temperature-continous}
    \mathtt{g}(t) = \frac{\partial T}{\partial n} \frac{1}{T}, 
  \end{equation}
   where the symbol $n$ defines the normal direction to the wall, bounds the viscous contribution to the time derivative of the entropy 
  \eqref{eq:continuousentropyestimate-1wall}.
\end{theorem}
\begin{proof}
Using the definition of matrices 
$\widehat{c}_{ij}$ given in \ref{app:cij}, the viscous vector-matrix-vector 
contraction given in 
\eqref{eq:continuousentropyestimate-1wall} yields the following term
\begin{equation}\label{eq:contraction-fv-continous}
  - \, w^{\top}\left(\widehat{c}_{11} \frac{\partial w}{\partial x_1} 
          + \widehat{c}_{12} \frac{\partial w}{\partial x_2}   
          + \widehat{c}_{13} \frac{\partial w}{\partial x_3} \right) 
          \:=\: 
          \kappa \left(\frac{\partial T}{\partial x_1}\frac{1}{T}\right).
\end{equation}
Therefore, specifying the condition $\mathtt{g}(t) = \left(\frac{\partial T}{\partial x_1} \frac{1}{T}\right)$
where $\mathtt{g}(t)$ is a known bounded function (i.e., $L^2 \cap L^{\infty}$),
eliminates the last potential source of instability on the right-hand-side of 
equation \eqref{eq:continuousentropyestimate-1wall}.
\end{proof}

The boundary condition given by Theorem 
\ref{th:continous-viscous-bound} at first glance appears ad hoc.
Note, however, that the scalar value 
$\kappa\left(\frac{\partial T}{\partial x_1} \frac{1}{T}\right)$ accounts for the 
change in entropy at the boundary $x_1 = 0$, due to the wall heat flux,
$\mathtt{q}_{wall}$, and it is often denoted
{\it{heat entropy transfer}} or {\it{heat entropy flow}}  
\cite{bejan-entropy-generation-book}.
In fact,
\begin{equation}\label{eq:origin-Tx-T}
\kappa \, \frac{\partial T}{\partial x_1} \frac{1}{T} = -\kappa \, \frac{\partial}{\partial
x_1}\left[w(5)\right] \frac{1}{w(5)} = \kappa \, \frac{\partial}{\partial x_1}
\left[\log(T) \right] = -\frac{\mathtt{q}_{wall}}{T_{wall}},
\end{equation}
where $w(5)$ denotes the fifth component of the entropy variable vector, $w$,
in \eqref{eq:entropy-variables-cnse}. 
Equation \eqref{eq:origin-Tx-T} indicates that, in the context of the entropy
stability analysis of the compressible Navier--Stokes equations, it
is admissible and physically (thermodynamically) correct to impose the fourth wall boundary condition as given in 
Theorem \ref{th:continous-viscous-bound}. 

To the best of our knowledge, Theorem 
\ref{th:continous-viscous-bound} provides new insight for future development of any
entropy stable boundary conditions for the compressible Navier--Stokes equations.

\begin{remark}
We strongly remark that the non-linear contraction obtained in
\eqref{eq:contraction-fv-continous} is different from that obtained 
from the linearized compressible Navier--Stokes equations
\cite{svard-no-slip-jcp-2008,berg2011}. The linear analysis produces velocity gradient terms
in the energy estimate (not present in \eqref{eq:contraction-fv-continous}), and temperature gradient terms
in the normal direction of the form 
$T \frac{\partial T}{\partial x_1}$, with $T$ and $\frac{\partial T}{\partial x_1}$
being perturbations; see, for instance, \cite{svard-no-slip-jcp-2008,berg2011}.
\end{remark}

\begin{remark}
The boundary condition $\mathtt{g}(t) = 0$, which corresponds to an adiabatic
wall $\partial T/\partial x_1 = 0$, bounds the solution,
and, as physically expected, does not contribute to the time derivative of 
the entropy \eqref{eq:continuousentropyestimate-1wall} because the heat flux is zero.
\end{remark}


\section{Entropy stable spectral discontinuous collocation method for 
the semi-discretized system: Interior operator}
\label{sec:inviscid-flux-condition}
In this section, we summarize the main results
that allow us to construct a numerical high order entropy stable discontinuous
collocation interior operator of any order, $p$, on unstructured grids. The formalism provided here, 
will then be used in Section \ref{sec:ss-no-slip-bc} to design new entropy 
stable solid wall boundary conditions for the semi-discretized compressible 
Navier--Stokes equations.

Using an SBP operator and its equivalent telescoping form (see,
for instance, \cite{ss-no-slip-wall-bc-parsani-nasa-tm-2014,CarpenterFisherSSDC2013NASA}), the 
semi-discrete form of the compressible Navier--Stokes equations \eqref{eq:cnse} 
becomes
\begin{equation}
\label{eq:nse:semidiscrete}
\begin{aligned}
  \frac{\partial \qN}{\partial t} & = \dmat_{x_i} \left(-{\bf f}_i^{(I)} +
  c_{ij} \dmat_{x_j} \qN\right)  + \pinv_{x_i} \left(\gN_i^{(B)} +
    \gN_i^{(In)}\right) \\
  & = \pinv_{x_i} \Delta_{x_i} \left( -\fM^{(I)}_i + \fM_i^{(V)} \right)
  +
  \pinv_{x_i} \left(\gN_i^{(B)} + \gN_i^{(In)}\right), \\
  &\hspace{-0.5cm} \qN(x,0) = \gN^{(0)}(x),\quad x \in \Omega,
\end{aligned}
\end{equation}
where the subscript $x_i$ indicates the coordinate direction to which the 
operators apply.
The source terms $\gN_i^{(B)}$ and $\gN_i^{(In)}$, with $i=1,2,3$, 
enforce boundary and interface conditions, respectively;
and $\gN^{(0)}$ represents the initial condition. The matrix 
$\pmat$ may be thought of as a mass matrix in the context of Galerkin 
finite element method. While it is not true in general that $\pmat$ is 
diagonal, herein the focus is exclusively on diagonal
norm SBP operators, based on fixed element-based polynomials of order $p$. 
The matrix $\dmat$ is used to approximate the first derivative; and it is
defined as $\pmat^{-1}\qmat$, where the nearly skew-symmetric matrix, $\qmat$, 
is an undivided differencing operator 
where all rows sum to zero and the first and last column sum to $-1$ and $1$, 
respectively \cite{ss-no-slip-wall-bc-parsani-nasa-tm-2014,CarpenterFisherSSDC2013NASA}. The operator $\Delta$ is the telescopic flux matrix and
allows to express the semi-discrete
system in a telescopic flux form, by evaluating the fluxes at the collocated 
flux points, $\fM^{(I)}$ and $\fM^{(V)}$\footnote{In the remainder of this work, all
quantities evaluated at the flux points are denoted with an overbar.}, (see Figure 
\ref{fig:stencil}). Note that the spacing between the flux points is
incorporated into the operator $\pmat$. In fact, the diagonal elements of $\pmat$
are equal to the spacing between flux points. In the remainder of this paper, 
the elements of $\pmat$ are denoted as $\pmat_{(i)(i)}, i=1,2,\ldots,N$.\footnote{$N=p+1$ for a $p$th-order
scheme.}
\begin{figure}[h]
 \centering
 \includegraphics[width=0.70\textwidth]{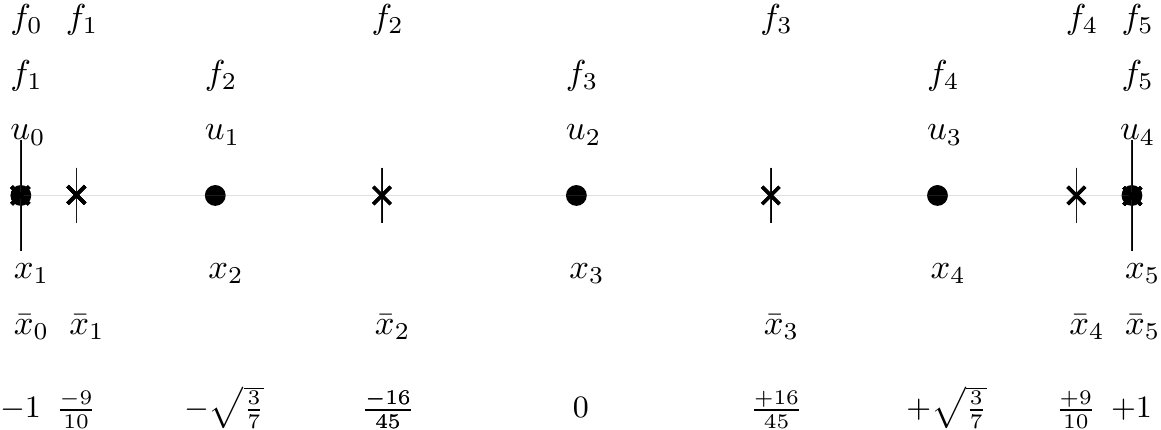}
 \caption{The one-dimensional discretization for $p=4$ Legendre collocation.
 Solution points are denoted by $\bullet$ and flux points are denoted by $\times$.}
 \label{fig:stencil}
\end{figure}
 
The semi-discrete entropy estimate is achieved by mimicking term by term the 
continuous estimate given in 
Equation \eqref{eq:entropyequation1}.
The non-linear analysis (entropy analysis) begins by contracting the semi-discrete equations given in 
Equation \eqref{eq:nse:semidiscrete} with the entropy variables: $\wN^{\top}\pmat$.
For clarity of presentation, but without loss of generality, the derivation is simplified 
to one spatial dimension. Tensor product algebra allows the results to extend directly to 
three dimensions. The resulting global equation that governs the
time derivative of the entropy is given by
\begin{equation}
  \label{eq:semidiscreteentropy}
  \wN^{\top} \pmat \frac{\partial \qN}{\partial t} + \wN^{\top} \Delta \fM^{(I)} =
  \wN^{\top} \Delta \fM^{(V)} + \wN^{\top} ( \gN^{(B)} + \gN^{(In)} ),
\end{equation}
where
\[
  \wN = \left(w(q_1)^{\top},w(q_2)^{\top},\dots,w(q_N)^{\top}\right)^{\top}
\]
is the vector of the entropy variables at the solution points.

\subsection{Time derivative}
The time derivative in \eqref{eq:semidiscreteentropy} is in
mimetic form for diagonal norm SBP operators.
The entropy variables are defined by the expression $w^{\top}=\partial
S/\partial q$ (see Definition \ref{def:entropyfunction}), which when combined with the
definition of entropy yields the point-wise expression
\[
  w_i^{\top} \frac{\partial q_i}{\partial t} \:=\: \frac{\partial S_i}{\partial
  q_i} \frac{\partial q_i}{\partial t} \:=\: \frac{\partial S_i}{\partial t}, \quad \forall i.
\]
Now, define the diagonal matrices $\partial {\bf S}/\partial {\bf q} \:=\: \WN \:=\: \textrm{diag}[\wN]$.
Since $\pmat$ is a diagonal matrix and arbitrary diagonal matrices commute, the semi-discrete rate of change
of entropy becomes
\[
  \wN^{\top} \pmat \frac{\partial\qN}{\partial t} 
\:=\: \mathbf{1}^{\top} \WN \pmat \frac{\partial\qN}{\partial t} 
\:=\: \mathbf{1}^{\top} \pmat \WN \frac{\partial\qN}{\partial t} 
\:=\: \mathbf{1}^{\top} \pmat \frac{\partial{\bf S}}{\partial \qN} \frac{\partial\qN}{\partial t}
\:=\: \mathbf{1}^{\top} \pmat \frac{\partial\SN}{\partial t},
\]
where  $$\mathbf{1} = \left(1,1,\dots,1\right)^{\top} $$
is a vector with $N$ elements.\footnote{Recall that $N=p+1$ for a $p$th-order
scheme.}

\subsection{Inviscid terms}
The inviscid portion of Equation \eqref{eq:semidiscreteentropy} is entropy 
conservative if it satisfies
\begin{equation}
\label{eq:entropyconsistent}
\wN^{\top} \Delta \fM^{(I)} = \overline{F}(q_N) - \overline{F}(q_1) = F(q_N) - F(q_1) = \mathbf{1}^{\top} \Delta \FM.
\end{equation}
Note that in \eqref{eq:entropyconsistent}, the first and
last flux points are coincident with the first and last solution points, which
enables the endpoint fluxes to be consistent (see Figure \ref{fig:stencil}).
Condition \eqref{eq:entropyconsistent}
leads to globally entropy conservative schemes but it is difficult to enforce 
at the flux points. One plausible solution to circumvent this problem is to use
a more restrictive point-wise relation between solution and flux-point data, which telescopes across the 
domain and yields a local condition on the flux for the global entropy
consistency constraint \eqref{eq:entropyconsistent}. Tadmor 
\cite{Tadmor2003} developed such a solution based on second-order accurate centered 
operators. Carpenter and co-authors \cite{CarpenterFisherSSDC2013NASA}, have 
generalized this solution for Legendre spectral collocation operators of any 
order of accuracy, $p$. 

In the following paragraphs, we present, without any proof, the main theorems 
which allow to construct inviscid entropy conservative and stable fluxes of any order
of accuracy, $p$. Interested readers should consult   
\cite{FisherCarpenter2013JCPb,CarpenterFisherSSDC2013NASA} and the references
therein for details. Note that in this section the subscripts 
$i-1$, $i$ and $i+1$ are used to denote a scalar or vector quantity at the $i-1$, 
$i$ or $i+1$ collocated point, and do not have to be confused with 
the subscript used, for instance, in \eqref{eq:cnse}.
\begin{theorem}
The local conditions
\begin{equation}
 \label{eq:entropyconsistentflux}
 \left(w_{i+1}-w_{i}\right)^{\top} \overline{f}_i = \widetilde{\psi}_{i+1}-\widetilde{\psi}_i, \quad i = 1,2,\dots,N-1  
\quad ; \quad \widetilde{\psi}_1 = \psi_1, \quad \widetilde{\psi}_N = \psi_N
\end{equation}
when summed, telescope across the domain and satisfy the entropy conservative
condition given in Equation \eqref{eq:entropyconsistent}.  
A flux that satisfies the condition given in equation \eqref{eq:entropyconsistentflux}
is denoted $\bar{f}_i^{(S)}$.  The potentials
$\widetilde{\psi}_{i+1}$ and $\widetilde{\psi}_i$ need not be the point-wise 
${\psi}_{i+1}$ and ${\psi}_i$, respectively.
\end{theorem}
\begin{proof}
See Theorem 3.3 in reference \cite{CarpenterFisherSSDC2013NASA} for the proof of
this theorem.
\end{proof}

A possible strategy for constructing high order entropy conservative fluxes 
is to construct a linear combination of two-point entropy
conservative fluxes by using the coefficients in the SBP matrix $\qmat$.  
This approach follows immediately from the generalized telescoping 
structural properties of diagonal norm SBP operators (see, for instance, 
\cite{ss-no-slip-wall-bc-parsani-nasa-tm-2014,CarpenterFisherSSDC2013NASA}). 
Because it requires only the existence of a two-point entropy conservative flux 
formula and the coefficients of $\qmat$, it is valid for any SBP operator 
that satisfies the following SBP constraints 
$$ 
\qmat^{\top} = \bmat - \qmat, \quad \bmat =
 \mathrm{diag}\left(-1,0,\dots,0,1\right).
$$
Thus, it is also valid for Legendre spectral collocation operators used herein;
see \cite{CarpenterFisherSSDC2013NASA}.

The following theorem establishes the accuracy of the new fluxes.
\begin{theorem}
\label{thm:hoentropyconsistent}
A two-point entropy conservative flux can be extended to high order with
formal boundary closures by using the form
\begin{equation}
\begin{gathered}
\label{eq:leflochfluxburgers2}
 \overline{f}_i^{(S)} = \sum_{k=i+1}^{N}\sum_{\ell = 1}^{i} 2 \qmat_{\ell k} \, \overline{f}_S\left(q_{\ell},q_{k}\right), \quad 1 \leq i \leq N-1,
\end{gathered} 
\end{equation}
when the two-point non-dissipative function from Tadmor \cite{Tadmor2003}, 
\begin{equation}
\label{eq:tadmorflux}
\overline{f}_S\left(q_{k},q_{\ell}\right) = 
\int_{0}^1 g\left(w(q_k) + \xi\left(w(q_\ell)-w(q_k)\right)\right) \dr{\xi},
\quad g(w(u)) = f(u),
\end{equation}
is used. 
The coefficient, $\qmat_{k \ell}$, corresponds to the
$k$ row and $l$ column in $\qmat$.
\end{theorem}
\begin{proof}
See Theorem 3.1 in reference \cite{FisherCarpenter2013JCPb} for the proof of
this theorem.
\end{proof}
Thus, Theorem \ref{thm:hoentropyconsistent} ensures that a high order flux 
constructed from a linear combination of two-point entropy 
conservative fluxes retains the design order
of the original discrete operator for any diagonal norm SBP matrix $\qmat$.

The following theorem establishes instead that the linear combination of 
two-point entropy conservative fluxes does preserve the property
of entropy stability for any arbitrary diagonal norm SBP matrix $\qmat$.
\begin{theorem}
\label{thm:entropyconsistency}
A two-point high order entropy conservative 
flux satisfying Equation \eqref{eq:entropyconsistentflux}
with formal boundary closures can be constructed using Equation \eqref{eq:leflochfluxburgers2},
\begin{equation}
\begin{gathered}
\label{eq:leflochfluxburgers3}
 \overline{f}_i^{(S)} = \sum_{k=i+1}^{N}\sum_{\ell = 1}^{i} 2 \qmat_{\ell k} \, \overline{f}_S\left(q_{\ell},q_{k}\right), \quad 1 \leq i \leq N-1,
\end{gathered}
\end{equation}
where $\bar{f}_S\left(q_\ell,q_k\right)$ is any two-point non-dissipative function that
satisfies the entropy conservation condition
\begin{equation}
\label{eq:2ptentropyconsistent}
\left(w_\ell-w_k\right)^{\top} \overline{f}_S\left(q_{\ell},q_{k}\right) = \psi_\ell - \psi_k.
\end{equation}
The high order entropy conservative flux satisfies an additional local
entropy conservation property,
\begin{equation}
\label{eq:localentropyconsistent1}
{\wN}^{\top} \pinv \Delta \fM^{(S)} = \pinv \Delta \FM = \frac{\partial
}{\partial x}F\left(\qN\right) + \tmat_{p+1},
\end{equation}
where $\tmat_{p+1}$ is the truncation error of the approximation of $\partial F\left(\qN\right)/\partial x$,
or equivalently,
\begin{equation}
\label{eq:localentropyconsistent}
w_{i}^{\top} \left(\overline{f}_{i}^{(S)} - \overline{f}_{i-1}^{(S)}\right)
= \left(\overline{F}_{i} - \overline{F}_{i-1}\right), \quad 1 \leq i \leq N,
\end{equation}
where
\begin{equation}
\label{eq:hoentropyflux}
\overline{F}_i = \sum_{k=i+1}^N \sum_{\ell = 1}^i
\qmat_{\ell k} \left[\left(w_\ell + w_k\right)^{\top} \overline{f}_S\left(q_{\ell},q_{k}\right)
- \left(\psi_{\ell}+\psi_k\right)\right], \quad 1 \leq i \leq N-1.
\end{equation}
\end{theorem}
\begin{proof}
See Theorem 3.2 in reference \cite{FisherCarpenter2013JCPb} for the proof of
this theorem.
\end{proof}

\subsubsection{Affordable entropy consistent Euler flux} \label{subsec:Ismail-Roe}
The inviscid terms in the discretization of the compressible Navier--Stokes 
equations \eqref{eq:nse:semidiscrete} are calculated according to Equations 
\eqref{eq:leflochfluxburgers3} by using
the two-point entropy conservative flux of Ismail and Roe \cite{Roe2009},
\begin{equation}
\label{eq:roeentropyconsistentflux}
\begin{gathered}
  \overline{f}_{S,j}(q_{i},q_{i+1}) = \left(
\hat{\rho}\hat{u}_j,
\hat{\rho}\hat{u}_j \hat{u}_1 + \delta_{j1} \hat{p},
\hat{\rho}\hat{u}_j \hat{u}_2 + \delta_{j2} \hat{p},
\hat{\rho}\hat{u}_j \hat{u}_3 + \delta_{j3} \hat{p},
\hat{\rho}\hat{u}_j \hat{H}
\right)^{\top}, \\
\hat{u} = \frac{\frac{\hat{u}_{i}}{\sqrt{T_{i}}}
+ \frac{\hat{\upsilon}_{i+1}}{\sqrt{T_{i+1}}}}
{\frac{1}{\sqrt{T_i}}+\frac{1}{\sqrt{T_{i+1}}}}, \quad
\hat{p} = \frac{\frac{\hat{p}_{i}}{\sqrt{T_{i}}}
+ \frac{\hat{p}_{i+1}}{\sqrt{T_{i+1}}}}
{\frac{1}{\sqrt{T_i}}+\frac{1}{\sqrt{T_{i+1}}}}, \\
\begin{aligned}
\hat{h} = R \frac{\log\left(\frac{\sqrt{T_{i}}\rho_i}{\sqrt{T_{i+1}}\rho_{i+1}}\right)}
{\frac{1}{\sqrt{T_i}} + \frac{1}{\sqrt{T_{i+1}}}}
& \left(\theta_1 + \theta_2\right), 
\end{aligned}
\\
\begin{aligned}
\theta_1 = \frac{\sqrt{T_i}\rho_i + \sqrt{T_{i+1}}\rho_{i+1}}
{\left(\frac{1}{\sqrt{T_{i}}} + \frac{1}{\sqrt{T_{i+1}}}\right)
\left(\sqrt{T_i}\rho_i-\sqrt{T_{i+1}}\rho_{i+1}\right)},
\end{aligned}
\\
\begin{aligned}
\theta_2 = \frac{\gamma+1}{\gamma-1}\frac{\log\left(\sqrt{\frac{T_{i+1}}{T_i}}\right)}
{\log\left(\sqrt{\frac{T_{i}}{T_{i+1}}}\frac{\rho_{i}}{\rho_{i+1}}\right)
\left(\frac{1}{\sqrt{T_i}}-\frac{1}{\sqrt{T_{i+1}}}\right)},
\end{aligned}
\\
\hat{H} = \hat{h} + \frac{1}{2}\hat{u}_\ell \hat{u}_\ell, \quad 
\hat{\rho} = \frac{\left(\frac{1}{\sqrt{T_{i}}} + \frac{1}{\sqrt{T_{i+1}}}\right)
\left(\sqrt{T_{i}}\rho_i - \sqrt{T_{i+1}}\rho_{i+1}\right)}
{2\left(\log(\sqrt{T_i}\rho_i) - \log(\sqrt{T_{i+1}}\rho_{i+1})\right)}. 
\end{gathered}
\end{equation}
The index $j$ denotes the spatial direction.
This somewhat complicated explicit form is the first entropy
conservative flux for the convective terms with low enough computational
cost to be implemented in a practical simulation code. 

To our knowledge, the Ismail and Roe flux \cite{Roe2009} 
cannot be written in the form given by \eqref{eq:tadmorflux}. 
Therefore, there is no mathematical proof that show that 
the entropy conservative flux constructed as in \eqref{eq:leflochfluxburgers3}
by using \eqref{eq:roeentropyconsistentflux}
will retain the design-order of the spatial discretization. However, thorough numerical
experiments reported in \cite{FisherCarpenter2013JCPb,carpenter-SSDC-siam-2014,
CarpenterFisher2013AIAA} and herein indicate that 
the inviscid terms calculated with the two-point entropy conservative flux of Ismail and Roe 
\cite{Roe2009} do not destroy the order of accuracy of the spatial operator.

\subsubsection{Entropy stable inviscid interface flux}
Herein, the solution between adjoining elements is allowed to be discontinuous. 
An interface flux that preserves the entropy consistency of the interior operators on either side
of the interface is needed.  An entropy consistent (or entropy conservative)
inviscid interface flux constructed according to equation
\eqref{eq:leflochfluxburgers3} by using 
\eqref{eq:roeentropyconsistentflux} is indicated as $\fs$. 
The superscripts $(-)$ and $(+)$ combined with the subscript $i$ denote 
the left and right state used to compute
the two-point entropy conservative flux and
therefore replace the subscripts $i$ and $i+1$ in  
\eqref{eq:roeentropyconsistentflux}.

A more 
dissipative and hence entropy {\it stable} inviscid interface flux 
$\fss$ is constructed as 
\begin{equation}
\label{eq-ssflux}
\fss \:=\: \fs \:+\: \Lambda \left(\wip-\wim\right),
\end{equation}
where $\Lambda$ is a negative semi-definite interface matrix with zero or negative eigenvalues. 
The entropy stable flux $\fss$ is more dissipative than the  
entropy conservative inviscid flux, as is
easily verified by contracting $\fss$ against the entropy variables to yield the expression 
\begin{equation}
 \label{eq:entropyconsistentfluxdiss1}
 \left(\wip-\wim\right)^{\top} \fss \:=\: \psip - \psim \:+\: 
 \left(\wip-\wim\right)^{\top} \Lambda \left(\wip-\wim\right).
\end{equation}
The matrix $\Lambda$ can be constructed using different approaches, e.g.,
using an upwind operator that dissipates each characteristic wave 
based on the magnitude of its eigenvalue: 
\begin{equation}
\label{eq-SSCharacteristic}
\begin{aligned}
  & \fssc \:=\: \fs \:+\: 1/2 \, {\Yeig} \, |\lambda| \, {\Yeig}^{\top}
  \left(\wim - \wip \right)  \: , \\
  & \frac{\partial}{\partial q}\mathbf{f}\left(q\right) \:=\: \Yeig \, \lambda \,  {\Yeig}^{\top}, \\
  & \frac{\partial q}{\partial w} \:=\: \Yeig {\Yeig}^{\top},
\end{aligned}
\end{equation}
where $\lambda$ and $\Yeig$ are the diagonal matrix of the eigenvalues and 
the matrix of the eigenvectors, respectively. 
Note that the relation $\frac{\partial q}{\partial w} \:=\: \Yeig {\Yeig}^T$ 
is achieved by an appropriate scaling of the
rotation eigenvectors. Unless otherwise
noted, the entropy stable characteristic flux \eqref{eq-SSCharacteristic} is used in all test simulations
presented herein. In particular, we adopt the scaled eigenvectors introduced by Merriam \cite{Merriam1989}
which allow to introduce an artificial viscosity from the
viewpoint of numerical satisfaction of the second law of thermodynamics.  

\begin{remark}
In \cite{svard-no-penetration-bc-JSC} grid interfaces for entropy stable finite
difference schemes are studied and interface fluxes similar to 
\eqref{eq-ssflux} are proposed.
\end{remark}

\subsection{Viscous Terms}
\label{sec:viscousterms}
Using the SBP formalism (see, for instance,
\cite{ss-no-slip-wall-bc-parsani-nasa-tm-2014,CarpenterFisherSSDC2013NASA}),
the contribution of the viscous terms to the semi-discrete time derivative
of the entropy is
\begin{equation}
\label{eq:ssvisc2}
\wN^{\top} \Delta \fM^{(V)} =  \wN^{\top} \bmat \, \widehat{c}_{11} \Dmat \wN 
- \left(\dmat \wN\right)^{\top} \pmat \, \widehat{c}_{11} \left(\dmat \wN\right) .
\end{equation}
The last term is negative semi-definite. 
As with the continuous estimate given in \eqref{eq:continuousentropyestimate}, only the boundary term 
can produce a growth of the entropy (see RHS of \eqref{eq:ssvisc2}), and thus the approximation of the
viscous terms is entropy stable. (Entropy stable boundary conditions bound these terms.)

\subsubsection{Entropy stable viscous interface coupling}
Herein, the 3D entropy stable viscous interface coupling procedure
proposed by Parsani, Carpenter and Nielsen \cite{ss-no-slip-wall-bc-parsani-nasa-tm-2014} is used 
to patch interior interfaces for the compressible Navier--Stokes equations.
This treatment is based on a precise combination of local discontinuous Galerkin (LDG) and 
interior penalty (IP) approaches.


\section{Entropy stable solid wall boundary conditions for the semi-discrete system}\label{sec:ss-no-slip-bc}
An estimate for the time derivative of the entropy of an isolated element is derived,
followed by a derivation of entropy stable penalty terms that impose physical data on a viscous wall.\footnote{The same boundary conditions (without stability proofs) 
  could be used for almost any spatial discretization,
  including the family of DG methods, FR approaches, WENO schemes, FD 
  and FV methods.} 

\subsection{General approach for the entropy stability analysis of a SBP-based spatial
discretization}
Consider a single tensor product element and a spatially discontinuous collocation
discretization with $N=p+1$ solution points in each coordinate direction; the following element-wise 
matrices will be used:
\begin{equation}\label{eq:SBP-tensor-matrices}
\begin{gathered}
  \dmat_{x_1} = \left(\dmat_{N} \otimes I_{N} \otimes I_{N} \otimes I_{5}\right), \quad \cdots \quad
  \dmat_{x_3} = \left(I_{N} \otimes I_{N} \otimes \dmat_{N} \otimes I_{5}\right), \\ \\
  \pmat_{x_1} = \left(\pmat_{N} \otimes I_{N} \otimes I_{N} \otimes I_{5}\right), \quad \cdots \quad
  \pmat_{x_3} = \left(I_{N} \otimes I_{N} \otimes \pmat_{N} \otimes I_{5}\right), \quad  \\ \\
  \pmat_{x_1 x_2} = \left(\pmat_{N} \otimes \pmat_{N} \otimes I_{N} \otimes I_{5}\right), \quad \cdots \quad
  \pmat_{x_2 x_3} = \left(I_{N} \otimes \pmat_{N} \otimes \pmat_{N} \otimes I_{5}\right), \\ \\
  \pmatv = \pmat_{x_1 x_2 x_3} = \left(\pmat_{N} \otimes \pmat_{N} \otimes
  \pmat_{N} \otimes I_{5} \right), \\ \\
  \bmat_{x_1} = \left(\bmat_{N} \otimes I_{N} \otimes I_{N} \otimes I_{5}\right), \quad  \cdots \quad
  \bmat_{x_3} = \left(I_{N} \otimes I_{N} \otimes \bmat_{N} \otimes I_{5}\right), \\ \\
  \Delta_{x_1} = \left(\Delta_{N} \otimes I_{N} \otimes I_{N} \otimes I_{5}\right), \quad\cdots \quad
  \Delta_{x_3} = \left(I_{N} \otimes I_{N} \otimes \Delta_{N} \otimes I_{5}\right), \\ \\
\end{gathered}
\end{equation}
where $\dmat_N$, $\pmat_{N}$, $\Delta_N$, and $\bmat_{N}$ are the
one-dimensional (1D) SBP operators \cite{ss-no-slip-wall-bc-parsani-nasa-tm-2014}, and  $I_{N}$ is
the identity
matrix of dimension $N$. $I_{5}$ denotes the 
identity matrix of dimension five.\footnote{The 3D compressible
  Navier--Stokes equations form a system of five non-linear PDEs.} 
The subscripts in \eqref{eq:SBP-tensor-matrices} indicate the coordinate 
directions to which the operators apply (e.g., $\dmat_{x_1}$ is 
the differentiation matrix in the $x_1$ direction). Furthermore, we define the 
norm $\wN^{\top} \pmatv \qN = \left\|{S}\right\|^2_{\pmatv}$, where $\wN$ and $S$ 
are the vector of the entropy variables at the solution points and the
mathematical entropy of the system, respectively.
When applying these operators to the scalar entropy equation in space,
a hat will be used to differentiate the scalar operator from the full vector 
operator. For example,
\begin{equation}
 \widehat{\pmatv} = \left(\pmat_{N} \otimes \pmat_{N} \otimes \pmat_{N} \right).
\end{equation}

Within one tensor product element, the 3D compressible 
Navier--Stokes equations are discretized as
\begin{equation}\label{eq:semi-discrete-element}
  \begin{aligned}
    \frac{\partial \qN}{\partial t} & 
    + \pmat^{-1}_{x_1}\Delta_{x_1} \left(\fM_{1}^{(I)} - \fM_1^{(V)} \right)
    + \pmat^{-1}_{x_2}\Delta_{x_2} \left(\fM_{2}^{(I)} - \fM_2^{(V)} \right)  
    + \pmat^{-1}_{x_3}\Delta_{x_3} \left(\fM_{3}^{(I)} - \fM_3^{(V)} \right) \\
    & = \pmat^{-1}_{x_1} \left(\gb_{1}^{(B)} + \gb_{1}^{(In)}\right)
    + \pmat^{-1}_{x_2} \left(\gb_{2}^{(B)} + \gb_{2}^{(In)} \right)  
    + \pmat^{-1}_{x_3} \left(\gb_{3}^{(B)} + \gb_{3}^{(In)} \right),
\end{aligned}
\end{equation}
where the vector of the conservative variables is ordered as
\begin{equation}\label{eq:ordered-unknown-vector-element}
  \qN = \left(q\left(x_{(1)(1)(1)}\right)^{\top}, 
              q\left(x_{(1)(1)(2)}\right)^{\top}, \ldots, 
              q\left(x_{(N)(N)(N)}\right)^{\top}  \right) =
              \left({q_{(1)}}^{\top},{q_{(2)}}^{\top}, \ldots, {q_{(N^3)}}^{\top}\right),
\end{equation}
and $\fM_{i}^{(I)}$ and $\fM_i^{(V)}, i=1,2,3,$ are the inviscid and viscous grid 
fluxes, respectively.\footnote{Recall that the vectors with an over-bar are defined at the 
flux points.} The vectors $\gb_{i}^{(B)}, \, i=1,2,3,$  
enforce the boundary conditions, while $\gb_{i}^{(In)}, i=1,2,3,$ patch 
interfaces together. The derivatives
appearing in the viscous fluxes are also computed using the
operator $\dmat_{x_i}, \, i=1,2,3,$ defined in \eqref{eq:SBP-tensor-matrices}. 

As in the continuous case, we apply the 
entropy analysis to Equation \eqref{eq:semi-discrete-element} by multiplying with 
$\wN^{\top} \pmatv$ from the left. 
Moreover, we substitute to $\fM_{i}^{(I)}, \, i=1,2,3,$ the 
high-order accurate entropy consistent flux constructed according to Equation 
\eqref{eq:leflochfluxburgers3} with the two-point
entropy conservative flux presented in Section \ref{subsec:Ismail-Roe}. The final 
expression for the time derivative of the entropy in the element is then 
\begin{equation}\label{eq:estimate-no-slip-bc-1}
  \begin{aligned}
    \frac{d}{d t} & \left\|{S}\right\|^2_{\pmatv}
    \:+\: \mathbf{1}^{\top} \left(\widehat{\pmat}_{x_2 x_3} \widehat{\bmat}_{x_1} \mathbfcal{\FM}_1 
                                + \widehat{\pmat}_{x_1 x_3} \widehat{\bmat}_{x_2} \mathbfcal{\FM}_2 
                                + \widehat{\pmat}_{x_1 x_2} \widehat{\bmat}_{x_3} \mathbfcal{\FM}_3  \right) \\
    & - \,\wN^{\top} \left(\pmat_{x_2 x_3} \bmat_{x_1} \fM_1^{(V)}
                           + \pmat_{x_1 x_3} \bmat_{x_2} \fM_2^{(V)}
                           + \pmat_{x_1 x_2} \bmat_{x_3} \fM_3^{(V)} \right) 
    + \, \mathbf{DT} \\
    & = \,\wN^{\top} \left(\pmat_{x_2 x_3} \left(\gb_{1}^{(B)} + \gb_{1}^{(In)}\right)
    + \pmat_{x_1 x_3} \left(\gb_{2}^{(B)} + \gb_{2}^{(In)} \right)  
    + \pmat_{x_1 x_2} \left(\gb_{3}^{(B)} + \gb_{3}^{(In)} \right)\right).
\end{aligned}
\end{equation}

Note that in \eqref{eq:estimate-no-slip-bc-1} the bar over the flux vectors
could be safely removed because the contraction of 
\eqref{eq:semi-discrete-element} against 
$\wN^{\top} \pmatv$ leads only to the fluxes at the face flux points, which are coincident 
with the first and last solution points (see Figure \ref{fig:stencil}). This 
duality is needed to define unique operators and is important in proving 
entropy stability \cite{CarpenterFisherSSDC2013NASA}.
The quantity $\mathbf{DT}$ denotes a positive quadratic term in the first
derivative approximation of the solution:
\begin{equation}\label{eq:DT-no-slip-bc}
\begin{aligned}
\mathbf{DT} & = \sum_{i=1}^3 \sum_{j=1}^3 \left(\dmat_{x_i} \wN\right)^{\top}
\pmatv [\chatmat_{ij}]  \left(\dmat_{x_j} \wN\right) \\
& =
\begin{pmatrix}
  \dmat_{x_1} \, \wN \\ \dmat_{x_2} \, \wN  \\ \dmat_{x_3} \, \wN
\end{pmatrix}^{\top}
\begin{pmatrix}
  \pmatv[\chatmat_{11}] & \pmatv[\chatmat_{12}] & \pmatv[\chatmat_{13}] \\
  \pmatv[\chatmat_{21}] & \pmatv[\chatmat_{22}] & \pmatv[\chatmat_{23}] \\
  \pmatv[\chatmat_{31}] & \pmatv[\chatmat_{32}] & \pmatv[\chatmat_{33}] \\
\end{pmatrix}
\begin{pmatrix}
  \dmat_{x_1} \, \wN \\ \dmat_{x_2} \, \wN  \\
  \dmat_{x_3} \, \wN
\end{pmatrix},
\end{aligned}
\end{equation}
where $[\chatmat_{ij}]$ denotes a block diagonal matrix with blocks
corresponding to the viscous coefficients of each solution point.
The positive semi-definiteness of $\mathbf{DT}$ follows from the
positivity of the matrices $\widehat{c}_{ij}$ (see Appendix B.2 in 
\cite{FisherCarpenter2013JCPb} for the proof and
\ref{app:cij} herein for the expression of these matrices). The matrices
$\bmat_{x_i}, \, i=1,2,3,$ pick the interface terms in the respective directions
(i.e., for a high-order accurate scheme on a tensor product cell, they pick the
solution value at the nodes of the two ``opposite" faces).
Therefore, Equation \eqref{eq:estimate-no-slip-bc-1} is the semi-discrete form 
of Equation \eqref{eq:continuousentropyestimate}, which was obtained from the
analysis at the continuous level. 

\subsection{Entropy stability analysis for the solid wall boundary conditions}
\label{subsec:entropy-stable-wall-bc}
In this section, we focus now on the construction of an entropy stable penalty term for imposing
the solid wall boundary conditions for the compressible Navier--Stokes
equations. 

Without loss of generality, we study a hexahedral element with edge length
equal to one and we consider only the face plane $(0,x_2,x_3)$.
With these assumptions, Equation 
\eqref{eq:estimate-no-slip-bc-1} reduces to
\begin{equation}\label{eq:estimate-no-slip-bc-x1}
  \begin{aligned}
    \frac{d}{d t} \left\|{S}\right\|^2_{\pmatv} & -  
    \mathbf{1}^{\top} \widehat{\pmat}_{x_2 x_3} \widehat{\gmat}_{(1)}
    \mathbfcal{\FM}_1
     + \,\wN^{\top} \pmat_{x_2 x_3} \gmat_{(1)} \fM_1^{(V)}
     + \, \mathbf{DT}  
     & = \,\wN^{\top} \pmat_{x_2 x_3} \gmat_{(1)} \gb_{1}^{(B)}.
\end{aligned}
\end{equation}
The operators $\widehat{\gmat}_{(k)}$ and $\gmat_{(k)}$ are defined as
\begin{equation}\label{eq:operator-gmat}
  \widehat{\gmat}_{(k)} = \left(\eN_{k} \otimes I_{N} \otimes I_{N}\right), \quad
  \gmat_{(k)} = \left(\eN_{k} \otimes I_{N} \otimes I_{N} 
    \otimes I_{5}\right),
\end{equation}
where $$\eN_{k} = \left(0,0,\ldots,1,0,\ldots,0,0\right)^{\top}$$ is a vector of
length $N$ and has a non-zero element corresponding to the location $k$. 
Therefore, the operators $\widehat{\gmat}_{(k)}$ and $\gmat_{(k)}$ 
 pick out the nodal values of the solution or any flux vector at a specific
 plane according to the ordering introduced in
 \eqref{eq:ordered-unknown-vector-element}.
Herein, the face plane $(0,x_2,x_3)$ is characterized by the index
$k=1$. Thus, Equation 
\eqref{eq:estimate-no-slip-bc-x1} represents the 
contribution to the time derivative of the entropy of 
the boundary points that lie on the face plane $(0,x_2,x_3)$. 

In the remainder of this paper, we assume that the node with solution vector 
$q\left(x_{(1)(1)(1)}\right)=q_{(1)}$
(see expression \eqref{eq:ordered-unknown-vector-element}) lies on this face plane. This point will be used to derive entropy stable wall boundary conditions.
All numerical states associated 
to it will be identified with the subscript $\left(\cdot\right)_{(1)}$.

In estimate \eqref{eq:estimate-no-slip-bc-x1}, the penalty source term $\gb^{(B)}_{1}$ is composed of 
three design-order terms that weakly enforce the wall boundary conditions:
\begin{equation}\label{eq:SAT-no-slip-bc}
  \begin{aligned}
    \gb_{1}^{(B)} =  - \left(
  \fb_1^{(I)} - \fsscb_1\left(\qN,\gb^{(E)} \right)\right) 
 + \left( \fM_1^{(V)} - \fM_1^{(V,B)}\right) 
 + [M] \left(\wN-\gb^{(NS),Vel}\right).
\end{aligned}
\end{equation}
In each of the three contributions, the first component (the numerical state) 
is constructed from the numerical solution, while the second component (the boundary state) 
is constructed from a combination of the numerical solution and 
four independent components of physical boundary data.

The first term enforces the Euler no-penetration wall 
condition through the inviscid flux of the compressible Euler equations.
The boundary state is formed by constructing an entropy conservative flux based on the 
numerical state at boundary point, $q_{(1)}$,
and a manufactured boundary state given by the 
vector $g^{(E)}$:
\begin{equation}\label{eq:definition-gE}
\begin{aligned}
  g^{(E)} = 
  \begin{pmatrix}
    1 & 0 & 0 & 0 & 0 \\
    0 &-1 & 0 & 0 & 0 \\
    0 & 0 & 1 & 0 & 0 \\
    0 & 0 & 0 & 1 & 0 \\
    0 & 0 & 0 & 0 & 1 
\end{pmatrix}
q_{(1)}^{\top} & = 
\left(\rho_{(1)},
-\left(\rho u_1\right)_{(1)},
\left(\rho u_2\right)_{(1)}, 
\left(\rho u_3\right)_{(1)}, 
\left(\rho E\right)_{(1)} \right)^{\top} \\
& = 
\left(q_{(1)}(1),
-q_{(1)}(2),
q_{(1)}(3), 
q_{(1)}(4), 
q_{(1)}(5) \right)^{\top}.
\end{aligned}
\end{equation}

The second term in \eqref{eq:SAT-no-slip-bc} enforces the heat entropy flow boundary condition \eqref{eq:temperature-continous},
facilitated by manufacturing a boundary viscous flux $\overline{f}_{1}^{(V,B)}$. 
Define the component of the gradient of the entropy variables in the numerical state as
\begin{subequations}
\begin{equation}
\Theta_{x_1, (1)} \:=\: \left[
  {\Theta}_{x_1, (1)}(1),
  {\Theta}_{x_1, (1)}(2),
  {\Theta}_{x_1, (1)}(3),
  {\Theta}_{x_1, (1)}(4),
  {\Theta}_{x_1, (1)}(5) \right]^{\top},
\end{equation}
\begin{equation}
\Theta_{x_2, (1)} \:=\: \left[
  {\Theta}_{x_2, (1)}(1),
  {\Theta}_{x_2, (1)}(2),
  {\Theta}_{x_2, (1)}(3),
  {\Theta}_{x_2, (1)}(4),
  {\Theta}_{x_2, (1)}(5) \right]^{\top},
\end{equation}
\begin{equation}
\Theta_{x_3, (1)} \:=\: \left[
  {\Theta}_{x_3, (1)}(1),
  {\Theta}_{x_3, (1)}(2),
  {\Theta}_{x_3, (1)}(3),
  {\Theta}_{x_3, (1)}(4),
  {\Theta}_{x_3, (1)}(5) \right]^{\top},
\end{equation}
\end{subequations}
where ${\Theta}_{x_i, (1)}(j)$ denotes the derivative of the $j$-$th$
entropy variable in the $i$ direction.
Next, specify the value of
$\mathtt{g}(t)$, 
the externally provided bounded function given by \eqref{eq:temperature-continous}.
Finally, define the manufactured component of the gradient in the normal
direction, $\widetilde{\Theta}_{x_1}$, as
\begin{equation}\label{eq:theta-tilde}
\widetilde{\Theta}_{x_1} \:=\: \left[
  {\Theta}_{x_1, (1)}(1),
  {\Theta}_{x_1, (1)}(2),
  {\Theta}_{x_1, (1)}(3),
  {\Theta}_{x_1, (1)}(4),
  \widetilde{\Theta}_{x_1}(5) \right]^{\top},
\end{equation}
where $\widetilde{\Theta}_{x_1}(5)$ is computed as
\begin{equation}\label{eq:gradient-viscous-bc-gradienT}
  \widetilde{\Theta}_{x_1}(5) = -\mathtt{g}(t) \,  w_{(1)}(5) =
  \frac{\mathtt{g}(t)}{T_{(1)}}.
\end{equation}
With these definitions, the manufactured viscous flux $\overline{f}_1^{(V,B)}$ is constructed as
\begin{equation}\label{eq:viscous-flux-bc}
  \overline{f}_{1}^{(V,B)} = \widehat{c}_{11} \, \widetilde{\Theta}_{x_1}
    + \widehat{c}_{12} \, \Theta_{x_2,(1)} + \widehat{c}_{13} \, \Theta_{x_3,(1)},
\end{equation}
where the matrices $\widehat{c}_{1j}$, $j=1,2,3$, are calculated using the numerical solution.
As we will show later, the boundary flux, $\overline{f}_{1}^{(V,B)}$,
constructed using \eqref{eq:viscous-flux-bc} 
will yield a mimetic contribution to the time
derivative of the entropy.
Note that for an adiabatic wall $\mathtt{g}(t) = 0$, and from expression 
\eqref{eq:gradient-viscous-bc-gradienT} we get $\widetilde{\Theta}_{x_1}(5) = 0$.

The third term in \eqref{eq:SAT-no-slip-bc} enforces the no-slip wall
(Dirichlet) boundary conditions ($u_1=u_2=u_3=0$) through 
a standard SAT approach. 
The manufactured boundary state $g^{(NS),Vel}$ is defined in terms of entropy variables as 
\begin{equation}\label{eq:def-data-NS}
  g^{(NS),Vel} = \left(w_{(1)}(1), 0, 0, 0, w_{(1)}(5)\right)^{\top}, 
\end{equation}
where, as usual, $w_{(1)}(1)$ and $w_{(1)}(5)$ are the first and the fifth components of 
the entropy vector constructed from the numerical solution.
Three boundary conditions are imposed in Equation \eqref{eq:def-data-NS};
all velocity components are set to zero at the wall. This is immediately clear by recalling that the 
entropy variables for the compressible Navier--Stokes equations are defined
as
$$
  w = \left( \frac{h}{T} - s - \frac{u_i u_i}{2T}, \frac{u_1}{T},
  \frac{u_2}{T},\frac{u_3}{T},-\frac{1}{T}\right)^{\top}.
$$
Note that the no-slip conditions are not used to define the first component of 
$g^{(NS),Vel}$, as the matrices $\widehat{c}_{ij},
  i,j=1,2,3$, have zeros on the first row and column (see
  \ref{app:cij}).\footnote{Currently, there are no diffusion terms in the
  equation that describes the conservation of mass.}
The matrix $[M]$ in \eqref{eq:SAT-no-slip-bc} is a block diagonal matrix 
with $N^3$ five-by-five blocks\footnote{$N^3=\left(p+1\right)^3$ is
the number of solution points within a three-dimensional tensor product cell.} which are defined as 
\begin{equation}\label{eq:IP-matrix-ss-wall-bc-1}
  M = -\frac{\alpha^{(B)}}{\left(\pmat_{x_1}\right)_{(1)(1)}} H \, 
  \widetilde{c}_{11} \, H, \quad H = \textrm{diag}(1,1,1,1,0), \quad
  \alpha^{(B)} > 0.
\end{equation}
The matrix $\widetilde{c}_{11}$ has the functional form of the usual 
symmetric positive semi-definite matrix $\widehat{c}_{11}$ defined in
\ref{app:cij}. This matrix has to be constructed using a set of primitive variables that is independent of
the numerical solution at all times. For example, for external flows, 
$\widetilde{c}_{11}$ can be constructed using the externally provided data at the far-field
(e.g., $(\rho_{\infty},|\vec{u}_{\infty}|,|\vec{u}_{\infty}|,
|\vec{u}_{\infty}|, T_{\infty})$).\footnote{In a general framework, the matrix $M$ is built
using the five-by-five matrix $\widetilde{c}_{ii}$ where the index $i$ 
denotes the normal direction to the wall.} 
The coefficient $\alpha^{(B)}$ in \eqref{eq:IP-matrix-ss-wall-bc-1} is used to modify
the strength of the SAT penalty term, and can be specified by the user. The factor 
$\left(\pmat_{x_1}\right)_{(1)(1)} > 0$ in the denominator is the first diagonal
element of the operator $\pmat_{x_1}$\footnote{Recall that
the diagonal element of any operator $\pmat$ are equal to the spacing between
flux points.}
and is introduced to get the correct dimensions. This factor is also
important because 
 it allows to achieve the correct asymptotic order of accuracy
 and yields an increase in the 
strength of $M$ with increased resolution.\footnote{This dependence
on the mesh size in the normal direction to the face is similar to that
of the interior penalty approach used in finite element methods.}


Summarizing Equation \eqref{eq:SAT-no-slip-bc}, the penalty at the face point is the sum
of three terms:
\begin{itemize}
  \item the difference between inviscid flux and the entropy consistent flux 
    at the node in the normal direction;
  \item the difference between the internal viscous flux and a boundary
    viscous flux at the node in the normal direction;
  \item the difference between the solution (in entropy variables) at the node
    and the data imposed at boundary, multiplied by the matrix $M$.
\end{itemize}

The penalty term \eqref{eq:SAT-no-slip-bc} contracted with the entropy
variables and simplified, yields the expression 
\begin{equation}\label{eq:penalty-bc-contribution-entropy}
\begin{aligned}
  RHS  = &- \wN^{\top}\pmat_{x_2 x_3} \gmat_{(1)} \left(
  \fb_1^{(I)} - \fsscb_1\left(\qN,\gb^{(E)} \right) \right) \\
  & + \wN^{\top} \pmat_{x_2 x_3} \gmat_{(1)} \left(
  \fM_1^{(V)} - 
\fM_1^{(V,B)}\right) \\
& + \wN^{\top} \pmat_{x_2 x_3}  \gmat_{(1)}[M]
  \left(\wN-\gb^{(NS),Vel}\right).
\end{aligned}
\end{equation}

The entropy stability of the penalty source term \eqref{eq:estimate-no-slip-bc-x1}
defined by Equation \eqref{eq:SAT-no-slip-bc} is demonstrated in the following
theorems. First, the inviscid term is proven to be entropy conservative and then
entropy stable, if dissipation is added. Next, the second term, which specifies 
the thermal condition, is proven to
be bounded by physical data provided by the user. Finally, the third term, which
specifies the no-slip boundary conditions, is proven to be entropy stable. 

\begin{theorem}\label{th:euler-flipping-sign}
The penalty inviscid flux contribution in Equation
\eqref{eq:SAT-no-slip-bc} is entropy conservative if the
vector $g^{(E)}$ is defined as in \eqref{eq:definition-gE}.
\end{theorem}
\begin{proof}
  The inviscid contribution of the boundary node, $\Upsilon^{(I)}$, to the 
time derivative of the entropy can be written as 
(see Equations \eqref{eq:estimate-no-slip-bc-x1} and 
\eqref{eq:penalty-bc-contribution-entropy})
\begin{equation}
\begin{aligned}
  \Upsilon^{(I)} & = \left(\pmat_{x_2 x_3}\right)_{(1)(1)} \overline{F}_1 - 
  w_1^{\top} \left(\pmat_{x_2 x_3}\right)_{(1)(1)} \left[\overline{f}_1\left(q_{(1)}\right) -
  f_1^{sr}\left(q_{(1)},g^{(E)}\right)\right],
\end{aligned}
\end{equation}
where $\left(\pmat_{x_2 x_3}\right)_{(1)(1)} \neq 0$.
Substituting the expression for the entropy flux $\overline{F}_1$ (i.e., 
Equation \eqref{eq:GodFpotential} with $i=1$) and evaluating the entropy 
consistent flux $f_1^{sr}$ using $q_{(1)}$ and $g^{(E)}$ yields the desired
result
\begin{equation}
  \begin{aligned}
    \Upsilon^{(I)} & = \left(\pmat_{x_2 x_3}\right)_{(1)(1)}\left[w_{(1)}^{\top}
    \overline{f}^{(I)}_1\left(q_{(1)}\right) - \psi_1 
    -  w_{(1)}^{\top} \overline{f}^{(I)}_1\left(q_1\right) + w_{(1)}^{\top}
  f_1^{sr}\left(q_{(1)},g^{(E)}\right)\right] \\
  & = \left(\pmat_{x_2 x_3}\right)_{(1)(1)}\left[- \psi_1 + w_{(1)}^{\top}
f_1^{sr}\left(q_{(1)},g^{(E)}\right)\right] = 0.
\end{aligned}
\end{equation}
\end{proof}

\begin{corollary}\label{cor:entropy-stable-inviscid-penalty}
The penalty inviscid flux contribution in Equation
\eqref{eq:SAT-no-slip-bc} is entropy stable if the
vector $g^{(E)}$ is defined as in \eqref{eq:definition-gE} and $\fsscb$ is
replaced by the entropy stable flux $\mathbf{f}^{ssr}$ defined in \eqref{eq-SSCharacteristic}.
\end{corollary}

\begin{remark}
A result similar to Corollary \ref{cor:entropy-stable-inviscid-penalty} is given by Sv{\"a}rd and {\"O}zcan 
\cite{svard-no-penetration-bc-JSC} in the context of high order entropy stable finite difference
schemes for the compressible Euler equations. Therein, an entropy dissipative
Euler no-penetration boundary treatment is proposed to bound the
time derivative of the entropy.
\end{remark}

Using Theorem \ref{th:euler-flipping-sign} we are left only with the viscous
contributions:
\begin{equation}\label{eq:estimate-no-slip-bc-viscous}
  \begin{aligned}
    \frac{d}{d t} \left\|{S}\right\|^2_{\pmatv} \:+\:  \wN^{\top}
    \pmat_{x_2 x_3} \, \gmat_{(1)} \,
    \fM_1^{(V)} +  \mathbf{DT}  
    \leq & +  \wN^{\top} \pmat_{x_2 x_3} \,  \gmat_{(1)} \left(
  \fM_1^{(V)}  -
\fM_1^{(V,B)} \right)  \\
    & + \wN^{\top}
    \pmat_{x_2 x_3} \, \gmat_{(1)} [M] \left(\wN-\gb^{(NS),Vel}\right).
\end{aligned}
\end{equation}

\begin{theorem}
  The viscous penalty terms in \eqref{eq:SAT-no-slip-bc},
  $$\gmat_{(1)} \left(
  \fM_1^{(V)} - \fM_1^{(V,B)}\right)  
 + \mathcal{G}_{(1)} [M] \left(\wN-\gb^{(NS),Vel}\right), $$ 
 are entropy
 stable for any value of $\mathtt{g}(t)$ and any five-by-five matrix 
$M$ as defined in \eqref{eq:IP-matrix-ss-wall-bc-1}.
\end{theorem}
\begin{proof}
  Clearly, the viscous flux term on the left-hand-side (LHS) of 
  \eqref{eq:estimate-no-slip-bc-viscous} is balanced by the same
  term on the RHS. Therefore, expression
\eqref{eq:estimate-no-slip-bc-viscous} reduces to
\begin{equation}\label{eq:estimate-no-slip-bc-viscous-2}
  \begin{aligned}
    \frac{d}{d t} \left\|{S}\right\|^2_{\pmatv} + \mathbf{DT}  
    \leq & - \wN^{\top} \pmat_{x_2 x_3} \, \gmat_{(1)} \, \fM_1^{(V,B)}  
    + \wN^{\top}
    \pmat_{x_2 x_3} \, \gmat_{(1)} \, [M] \left(\wN-\gb^{(NS),Vel}\right).
\end{aligned}
\end{equation}
  The contraction $-\wN^{\top} \pmat_{x_2 x_3}  \, \gmat_{(1)} \,
  \fM_1^{(V,B)}$ with $\overline{f}_1^{(V,B)}$ defined as in \eqref{eq:viscous-flux-bc}
  yields the following nodal contribution
  \begin{equation}\label{eq:mimetic-fvbc-1}
   -w_{(1)}^{\top} \left(\pmat_{x_2
   x_3}\right)_{(1)(1)}  \overline{f}_1^{(V,B)} = \left(\pmat_{x_2
   x_3}\right)_{(1)(1)} \, \kappa \, \mathtt{g}(t).
  \end{equation} 
  Since $\mathtt{g}(t)$ is a known bounded function (i.e., $L^2 \cap L^{\infty}$)
  expression \eqref{eq:mimetic-fvbc-1} 
  is also bounded.
  We highlight that for an adiabatic wall $\mathtt{g}(t)=0$ and,
  consequently, the viscous flux penalty in \eqref{eq:SAT-no-slip-bc} conserves
  the entropy (as it should) because the heat flux is zero.\footnote{$\mathtt{g}(t)=0$ in Equation
  \eqref{eq:mimetic-fvbc-1} yields zero.}
  Note that the contribution \eqref{eq:mimetic-fvbc-1} 
  mimics exactly the boundary contribution to the time derivative of the entropy
  that has been obtained from the 
  continuous analysis (see Equation \eqref{eq:contraction-fv-continous}).

  We are then left with the contribution 
  $\wN^{\top} \pmat_{x_2 x_3}\mathcal{G}_{(1)} \, [M] \left(\wN-\gb^{(NS),Vel}\right)$. 
  At the nodal level, this term can be re-written as
  \begin{equation}
    \begin{aligned}\label{eq:SAT-no-slip-rewritten}
      w_{(1)}^{\top} \left(\pmat_{x_2 x_3}\right)_{(1)(1)} M \left(w_{(1)}-g^{(NS),Vel}\right)
       = &
       + \frac{1}{2} \, \left(\pmat_{x_2 x_3}\right)_{(1)(1)} \, w_{(1)}^{\top} M
    w_{(1)} \\
    &
    - \frac{1}{2} \, \left(\pmat_{x_2 x_3}\right)_{(1)(1)} \, \left(g^{(NS),Vel}\right)^{\top}  M
      \,  g^{(NS),Vel} \\ 
      & + \frac{1}{2} \, \left(\pmat_{x_2 x_3}\right)_{(1)(1)} \, \left(w_{(1)} - g^{(NS),Vel} \right)^{\top} M \left(w_{(1)} -
  g^{(NS),Vel} \right).
  \end{aligned}
  \end{equation}
The penalty contribution given by Equation \eqref{eq:SAT-no-slip-rewritten} 
imposes the no-slip Dirichlet boundary conditions on the velocity components and is bounded if
\begin{itemize}
  \item $M$ is negative semi-definite;
  \item $M$ is independent of the numerical state.
\end{itemize}
If these two conditions are fulfilled, the first and the last term 
in \eqref{eq:SAT-no-slip-rewritten} introduce only dissipation, whereas the 
second one is a bounded term because it is just a function of data; and it is zero for no-slip boundary conditions.
\end{proof}

For a Reynolds number ($Re$) that approaches $+\infty$, 
we would like to smoothly recover only the no-penetration (or wall slip) boundary condition that 
characterizes the Euler equations (first contribution in 
\eqref{eq:SAT-no-slip-bc}). To achieve that, the matrix $M$ needs to be a
function of the Reynolds number and can be computed as in \eqref{eq:IP-matrix-ss-wall-bc-1},
i.e., 
\begin{equation*}
  M = -\frac{\alpha^{(B)}}{\left(\pmat_{x_1}\right)_{(1)(1)}} H \, 
  \widetilde{c}_{11} \, H, \quad H = \textrm{diag}(1,1,1,1,0), \quad
  \alpha^{(B)} > 0,
\end{equation*}
where $\widetilde{c}_{11}$ has the functional form of the usual 
$\widehat{c}_{11}$ matrix and it is constructed using a state that is
independent of the numerical solution at all times.


\section{Numerical results}\label{sec:validation}
The objective of this section is to demonstrate the accuracy 
and robustness of the new entropy stable wall boundary conditions coupled 
with the family of high-order entropy stable interior operators developed in 
\cite{CarpenterFisherSSDC2013NASA}. The unstructured grid solver used herein 
uses a transformation 
from computational to physical space that satisfies the semi-discrete geometric 
conservation law.

Before proceeding with the numerical tests,
we demonstrate with an example that the construction of a penalty source term with only
an inviscid and a viscous contributions leads to a non-entropy stable boundary treatment.

\subsection{Non-entropy stable viscous wall boundary conditions: Isothermal
wall}
In Section \ref{subsec:entropy-stable-wall-bc}, we have shown that 
constructing  $\gb^{(B)}_{1}$ as in \eqref{eq:SAT-no-slip-bc} 
yields entropy stable wall boundary conditions. However, one might attempt to 
construct $\gb^{(B)}_{1}$ as the sum of an inviscid penalty flux and only a viscous interior
penalty term, 
\begin{equation}\label{eq:SAT-no-slip-bc-wrong}
  \begin{aligned}
    \gb_{1}^{(B)} = & -\gmat_{(1)}\left(
  \fb_1^{(I)} - 
\fsscb_1\left(\qN,\gb^{(E)} \right)\right)  
+ \mathcal{G}_{(1)} \, [L] \left(\wN-\gb^{(NS)}\right).
\end{aligned}
\end{equation}
For an isothermal wall, for instance, $\gb^{(NS)}$\footnote{Note that $\gb^{(NS)}$ is expressed in terms of entropy variables.} is a vector of data 
that imposes 
both the no-slip boundary conditions (i.e.,
$u_1=u_2=u_3=0$) and the wall temperature:
\begin{equation}\label{eq:def-data-NS-wrong}
  g^{(NS)} = \left(w_{(1)}(1), 0, 0, 0, -\frac{1}{T_{wall}}\right)^{\top}. 
\end{equation}
The matrix $[L]$ in \eqref{eq:SAT-no-slip-bc-wrong} is a block diagonal matrix with $N^3$ blocks of size 
five-by-five.
Comparing the two definitions of $\gb_{1}^{(B)}$ given in Equations 
\eqref{eq:SAT-no-slip-bc} and \eqref{eq:SAT-no-slip-bc-wrong}, it can be seen that
in the latter approach  no 
viscous flux penalty terms are introduced. This is a key difference,
and as shown in \ref{app:counter-example}, yields
a provably non-entropy stable solid wall boundary conditions. Such a boundary treatment leads to
unstable simulations when used in combination with fine grids and/or high-order accurate polynomial 
representations of the solution.

\subsection{Computation of a square cylinder in subsonic freestream}
The flow past a square cylinder represents a benchmark test case for external 
flow past bluff bodies. This flow has been the subject of intense
experimental and numerical research in the past. In fact, a cylinder with square
cross section is a 
simple but a central shape for many engineering applications, including 
aeroacoustics and air pollutant transport and dispersion in urban environments.

The flow is described in a Cartesian coordinate system ($x_1,x_2,x_3$), in which
the $x_1$-axis is aligned with the inlet flow direction, the $x_3$-axis is
parallel with the cylinder axis and the $x_2$-axis is perpendicular to both
directions (see Figure \ref{fig:square-grid-re200-ma01}). A fixed two-dimensional square cylinder with
a side $d$ is exposed to a uniform freestream velocity vector with modulus
$|\vec{u}_{\infty}|$. The length
of the square cylinder in the $x_3$-direction is $10\,d$.

The following boundary 
conditions are used. A uniform flow
is prescribed at the inlet which is located $10\,d$ units upstream of the
cylinder. At the outlet, located $20\,d$ unit downstream of the cylinder,
far-field boundary conditions are used. A no-penetration (Euler) boundary 
condition is prescribed at the upper and lower boundaries. No-slip and adiabatic
conditions are enforced at the body surface. 
A periodic boundary condition is
used in the spanwise direction $x_3$. In the $x_2$-direction, the solid
blockage of the confined flow (i.e., the vertical distance between the upper and
the lower inviscid walls) is set to $18\,d$. 

The flow has a freestream Mach number of $M_{\infty} = 0.1$, and a Reynolds
number of $Re_{\infty} =2\times10^2$. The Reynolds number is based on the
modulus of
the freestream velocity vector, $|\vec{u}_{\infty}|$ and the height of the cylinder $d$. At
this Reynolds number, the regime is laminar and it usually persists up to
a Reynolds number of about $4\times 10^2$. Moreover, the vortex shedding is 
characterized by one very well-defined frequency
\cite{square-cylinder-norberg-1993}. A very
small time step is used to integrate the system of ordinary differential
equations (ODEs) so that the temporal 
error is negligible compared to that of the spatial discretization.

\subsubsection{Accuracy of the no-slip wall boundary conditions}
The proposed entropy stable no-slip wall boundary conditions do not force the
numerical solution to exactly fulfill the boundary conditions. Instead the
effect can be described as a rubber-band pulling the solution towards the
boundary conditions. The computed boundary value (or numerical state) typically deviates slightly 
from the prescribed value but the deviation is reduced as the grid is refined.
Therefore, the error at the boundary
can serve as a rough measure of the error of the entire solution. 

We compute the maximum norm $L^{\infty}$ of the error of the three velocity
components $u_1$, $u_2$, and $u_3$ on the complete surface of the cylinder at $t=1$, 
for three different grids.
The
meshes are fully unstructured, although a structured subdivision is used around 
the square cylinder and the near wake region to perform a grid convergence study 
(see Figure \ref{fig:square-grid-re200-ma01}). Grid 3 is the finest grid and 
has $20$ points on each side of the square, $20$ points in the ``radial" 
direction in the ``structured portion'' near the body, $40$ points in the 
near wake region in the freestream direction, and 
$8$ points in the spanwise direction. Grid 2 and Grid 1 are obtained by taking
every other and every fourth grid point of Grid 3 in the structured region. 
The simulations are performed using  different
orders of the polynomial ($p = 1,2,3,4$). The results are shown in Tables 
\ref{tbl:error-u1-square-re200-ma01}, \ref{tbl:error-u2-square-re200-ma01}, and
\ref{tbl:error-u3-square-re200-ma01}.  
 \begin{figure}[h]
 \centering
 \includegraphics[width=0.70\textwidth]{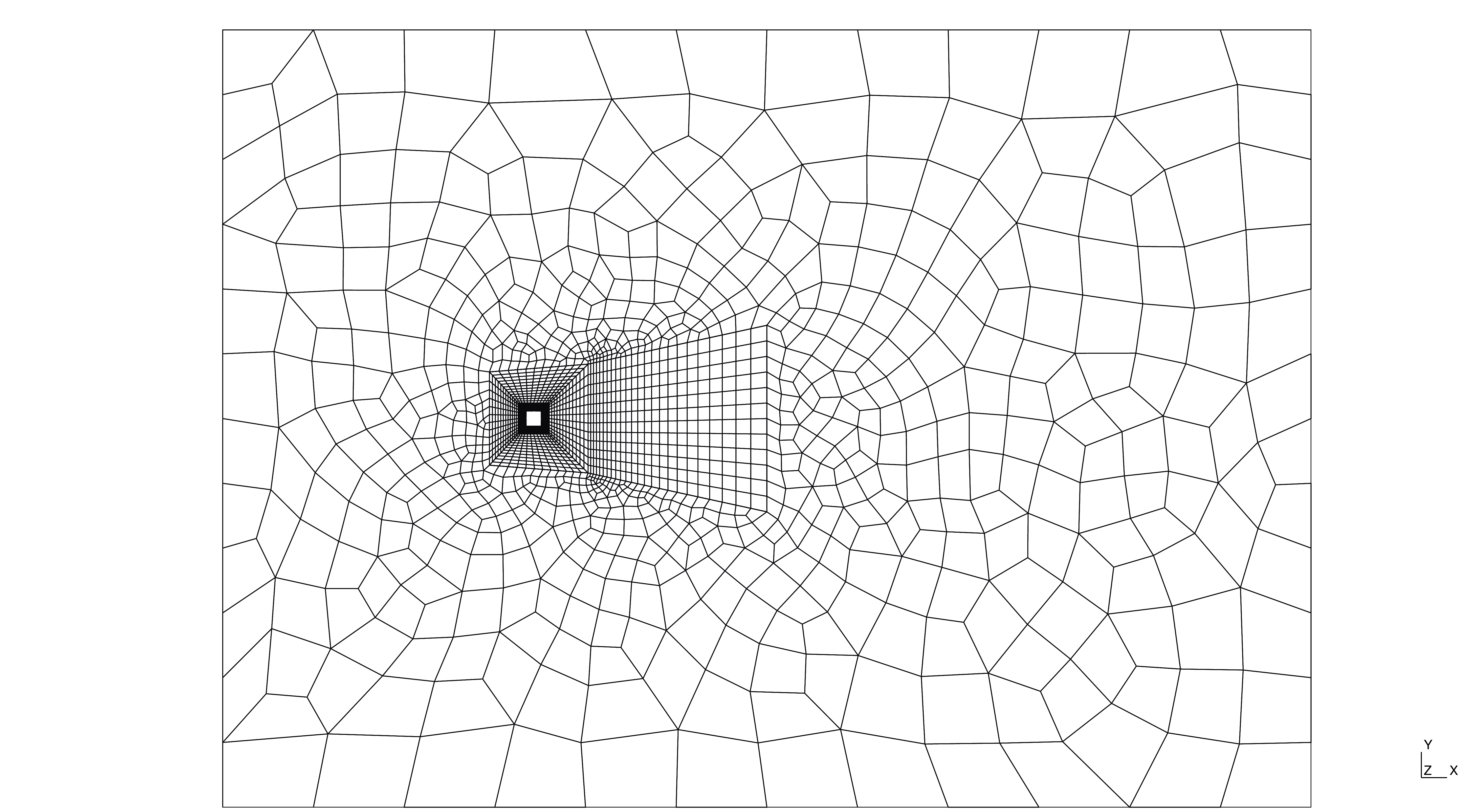}
 \caption{Example of structured/unstructured grids used for the flow past a 3D 
   square cylinder at $Re_{\infty}=2\times10^2$ and $M_{\infty}= 0.1$.}
 \label{fig:square-grid-re200-ma01}
\end{figure}
 
 We highlight a few observations. First, in
 all cases an increase in theoretical order of accuracy results in a error 
 reduction on all grids. Secondly, although the convergence rates in
 model problems are shown on much finer meshes, the 
 computed order of
 accuracy is very close to the formal value between the medium and fine 
 grids, even for these more realistic meshes.
 \begin{table}[htbp]
  \centering
\begin{tabular}{c| c c | c c | c c | c c}  
           & $p=1$        & rate & $p=2$        & rate  & $p=3$         & rate
  & $p=4$ & rate \\ \hline
  Grid 1            & 4.73e-2  &  -   & 2.15e-2  &  -    & 9.61e-3   &  -    &
  4.54e-3  &  -   \\  
  Grid 2            & 1.47e-2  & 1.69 & 2.88e-3  & 2.90   & 5.83e-4  & 4.04  &
  1.39e-4  &  5.02 \\ 
  Grid 3            & 3.55e-3  & 2.04 & 3.66e-4  & 2.98   & 3.40e-5  & 4.10  &
  4.52e-6  &  4.94  \\
\end{tabular}
   \caption{$L^{\infty}$ error norm of the velocity component $u_1$ at the wall
   and convergence rates; $t =1$; 3D unsteady laminar flow past a square cylinder at $Re_{\infty}=2\times10^2$ and $M_{\infty} =
  0.1$.} 
    \label{tbl:error-u1-square-re200-ma01}
\end{table}

\begin{table}[htbp]
  \centering
\begin{tabular}{c| c c | c c | c c | c c}  
          & $p=1$        & rate & $p=2$        & rate  & $p=3$        & rate &
  $p=4$ & rate \\ \hline
  Grid 1            & 7.20e-2  &  -   & 2.71e-2  &  -    & 1.43e-2  &  -   &
  5.14e-3  &  -   \\  
  Grid 2            & 1.79e-2  & 2.01 & 3.34e-3  & 3.02  & 1.10e-3  & 3.70 &
  1.96e-4  &  4.71 \\ 
  Grid 3            & 4.65e-3  & 1.94 & 4.20e-4  & 2.99  & 7.21e-5  & 3.93 &
  6.48e-6  &  4.92  \\
\end{tabular}
   \caption{$L^{\infty}$ error norm of the velocity component $u_2$ at the wall
   and convergence rates; $t =1$; 3D unsteady laminar flow past a square cylinder at $Re_{\infty}=2\times10^2$ and $M_{\infty} =
  0.1$.} 
    \label{tbl:error-u2-square-re200-ma01}

\end{table}

\begin{table}[htbp]
  \centering
\begin{tabular}{c| c c | c c | c c | c c}  
           & $p=1$        & rate  & $p=2$         & rate  & $p=3$        & rate
  & $p=4$ & rate \\ \hline
  Grid 1            & 2.75e-4  &  -    & 1.34e-4   &  -    & 1.01e-4  &  -   &
  8.62e-5  &  -    \\  
  Grid 2            & 5.98e-5  & 2.20  & 1.71e-5   & 2.97  & 7.92e-6  & 3.67 &
  3.14e-6  & 4.78  \\ 
  Grid 3            & 1.38e-5  & 2.12  & 2.03e-6   & 3.04  & 5.30e-7  & 3.90 &
  8.70e-8  & 5.17  \\
\end{tabular}
   \caption{$L^{\infty}$ error norm of the velocity component $u_3$ at the wall
   and convergence rates; $t =1$; 3D unsteady laminar flow past a square cylinder at $Re_{\infty}=2\times10^2$ and $M_{\infty} =
  0.1$.} 
    \label{tbl:error-u3-square-re200-ma01}
\end{table}

\subsubsection{Vortex shedding}
In this section we investigate the vortex shedding and the time variation of the
lift and drag coefficients. We compare our results against the data reported by
 Sohankar et al. \cite{sohankar-square-1999}. 
We compute the
following quantities: The Strouhal number, $f\,d/|\vec{u}_{\infty}|$, where $f$
is the frequency of the vortex shedding; the time-averaged drag coefficient, 
$c_D$, and the spanwise-averaged root-mean-square (RMS) of the
lift coefficient, $c_{L}^{RMS}$. We use the same grids
presented in the previous section, and different orders of the polynomial 
($p = 1,2,3,4$). The results are illustrated in Tables 
\ref{tbl:st-cd-cl-square-re200-ma01-coarse}, 
\ref{tbl:st-cd-cl-square-re200-ma01-medium}, and
\ref{tbl:st-cd-cl-square-re200-ma01-fine}.
From these tables, it can be seen that in all cases the
accuracy of the results improve by increasing the order of accuracy of the scheme.
We also note that, on Grid 3, which is very coarse compared to the typical grids
used with second-order FV and FD schemes, fourth- ($p=3$) and fifth-order ($p=4$) accurate
entropy stable schemes perform very well. In fact, the 
aerodynamic coefficients computed with these two discretizations are in very good
agreement with the results reported in literature \cite{sohankar-square-1999}.

\begin{table} 
  \centering
\begin{tabular}{c| c c c }  
  Solution                                     & $St$   & $\langle c_D \rangle$ &  $c_{L}^{RMS}$ \\ \hline
  SSDC $p=1$                                   & 0.098  & 1.01                  &  0.02           \\
  SSDC $p=2$                                   & 0.109  & 1.08                  &  0.06           \\
  SSDC $p=3$                                   & 0.142  & 1.19                  &  0.11           \\
  SSDC $p=4$                                   & 0.151  & 1.28                  &  0.15           \\
  Sohankar et al. \cite{sohankar-square-1999}  & 0.160  & 1.41                  &  0.22           \\  
\end{tabular} 
  \caption{Strouhal number, mean drag coefficient, and spanwise-averaged RMS of the lift coefficient for the 3D 
    unsteady laminar flow past a square cylinder at $Re_{\infty}=2\times10^2$ and $M_{\infty} =
    0.1$; Grid 1.} 
    \label{tbl:st-cd-cl-square-re200-ma01-coarse}
\end{table}

\begin{table} 
  \centering
\begin{tabular}{c| c c c }  
  Solution                                     & $St$   & $\langle c_D \rangle$ &  $c_{L}^{RMS}$ \\ \hline
  SSDC $p=1$                                   & 0.128  & 1.16                  &  0.07          \\
  SSDC $p=2$                                   & 0.139  & 1.28                  &  0.13           \\
  SSDC $p=3$                                   & 0.153  & 1.36                  &  0.20           \\
  SSDC $p=4$                                   & 0.159  & 1.40                  &  0.23           \\
  Sohankar et al. \cite{sohankar-square-1999}  & 0.160  & 1.41                  &  0.22           \\  
\end{tabular} 
   \caption{Strouhal number, mean drag coefficient, and spanwise-averaged RMS of the lift coefficient for the 
    3D unsteady laminar flow past a square cylinder at $Re_{\infty}=2\times10^2$ and $M_{\infty} = 0.1$; 
    Grid 2.} 
    \label{tbl:st-cd-cl-square-re200-ma01-medium}
\end{table}

\begin{table} 
  \centering

\begin{tabular}{c| c c c }  
  Solution                                     & $St$   & $\langle c_D \rangle$ &  $c_{L}^{RMS}$ \\ \hline  
  SSDC $p=1$                                   & 0.134  & 1.29                  &  0.12           \\
  SSDC $p=2$                                   & 0.154  & 1.37                  &  0.19           \\
  SSDC $p=3$                                   & 0.159  & 1.40                  &  0.22           \\
  SSDC $p=4$                                   & 0.159  & 1.42                  &  0.23           \\
  Sohankar et al. \cite{sohankar-square-1999}  & 0.160  & 1.41                  &  0.22           \\
\end{tabular}
   \caption{Strouhal number, mean drag coefficient, and spanwise-averaged RMS of the lift coefficient for the 
    3D unsteady laminar flow past a square cylinder at $Re_{\infty}=2\times10^2$ and $M_{\infty} = 0.1$; 
    Grid 3.} 
    \label{tbl:st-cd-cl-square-re200-ma01-fine}
\end{table}

\subsection{Heat entropy flow}

In this section we perform a convergence study of the thermal condition to
verify that the heat entropy flow at the wall converges to the prescribed
value. At $t=100$, we compute the maximum norm $L^{\infty}$ of the error of the quantity
$\left[\widetilde{\Theta}_{x_1}(5) T_{(\cdot)}\right]$ (see expression \eqref{eq:gradient-viscous-bc-gradienT}) 
for all the solution points that lie on the solid wall.
The entropy flow $g(t)$ is set to $\mathtt{g}(t)=const=0.02$. The results are illustrated in Table
\ref{tbl:error-g-square-re200-ma01}.
\begin{table}[htbp]
  \centering
\begin{tabular}{c| c c | c c | c c | c c}  
          &   $ p=1$      & rate  &  $p=2$     & rate   & $p=3$     & rate  & $p=4$   & rate \\ \hline
  Grid 1  &   6.17e-2     & -     &  3.68e-2   &  -     & 9.16e-3   &  -    & 7.23e-3 &  -   \\  
  Grid 2  &   3.35e-2     & 0.88  &  9.95e-3   &  1.89  & 1.32e-3   & 2.79  & 5.83e-4 & 3.63 \\ 
  Grid 3  &   1.63e-2     & 1.03  &  2.23e-3   &  2.16  & 1.68e-4   & 2.97  & 2.59e-5 & 4.22 \\
\end{tabular}
   \caption{$L^{\infty}$ error norm of the heat entropy flow at the wall
   and convergence rates; $\mathtt{g}(t) = const. = 0.02$ (see
   \eqref{eq:gradient-viscous-bc-gradienT}); $t =100$; 3D unsteady laminar flow past a square cylinder at $Re_{\infty}=2\times10^2$ and $M_{\infty} =
  0.1$.} 
    \label{tbl:error-g-square-re200-ma01}

\end{table}
As for the convergence study on the no-slip boundary conditions, it can be seen that in all cases the
accuracy of the results improve by increasing the order of accuracy of the
scheme. On this set of grids, the convergence rate of the error associated to the
heat entropy flow is $p$.

\subsection{Computation of a square cylinder in supersonic freestream}
The development of a high-order accurate entropy stable discretization 
aims to provide the next generation of robust high fidelity numerical solvers 
for complex fluid flow simulations, for which standard suboptimal algorithms
suffer greatly or fail completely. By computing the flow past a 3D square cylinder 
at $Re_{\infty}=10^4$ and $M_{\infty}=1.5$, we provide numerical evidence of
such robustness for the complete entropy stable high order spatial
discretization. This supersonic flow
is characterized by a very large range of length scales, strong shocks and
expansion regions that interact with each other, leading to complex flow
patterns. During the past three decades, this fluid flow problem has been 
thoroughly investigated by several researchers for
aerodynamic applications (see, for instance,  
\cite{supersonic-square-cylinder-nakagawa-1986,
supersonic-square-cylinder-nakagawa-1987,square-cylinder-birch-2003}).

The domain of interest spans one square cylinder edge in the $x_3$ direction, and 
at the two planes perpendicular to this coordinate direction, periodic
boundary conditions are used. 
The flow is computed using an unstructured grids with 
$43,936$ hexahedrons. A
fourth-order accurate ($p=3$) entropy stable discretization without 
absolutely any 
stabilization technique is used. The body surface is considered
adiabatic. The solution is initialized using a 
uniform flow at $M_{\infty}=1.5$ with zero angle of attack.

At the beginning of the simulation a strong shock is
formed in front of the bluff body. Subsequently, the discontinuity moves 
upstream until it reaches a ``stationary'' position that is about $2.15$ square 
cylinder edges far from the frontal surface of the body. 
During this phase, additional weaker shocks, which originate from the four sharp corners 
of the body,
interact with the subsonic regions formed near the walls. This 
complicated flow pattern, yields the formation of shock-lets in the wake
of the square cylinder. Figure \ref{fig:supersonic-square-1}
shows a portion of the ``high order grid''\footnote{Original grid with
element-wise interior node connections.} close to the body and its near-wake
region, and the Mach number and 
density contours at $t=1.5$. It can be seen
that relatively small oscillations are generated in front of the shock. This
numerical feature is absolutely natural and expected because the solution has been
computed with a fourth-order
accurate scheme without artificial dissipation or 
filtering technique. Nevertheless, the simulation
remains stable at all time, and the oscillations are always confined in 
small regions close to the discontinuities.
\begin{figure}[htbp]
 \centering
 \subfigure[High order grid in the near-body and near-wake regions.]{
  \includegraphics[trim=0.5cm 2.5cm 0.5cm 2.5cm, clip=true, totalheight=0.29\textheight]
  {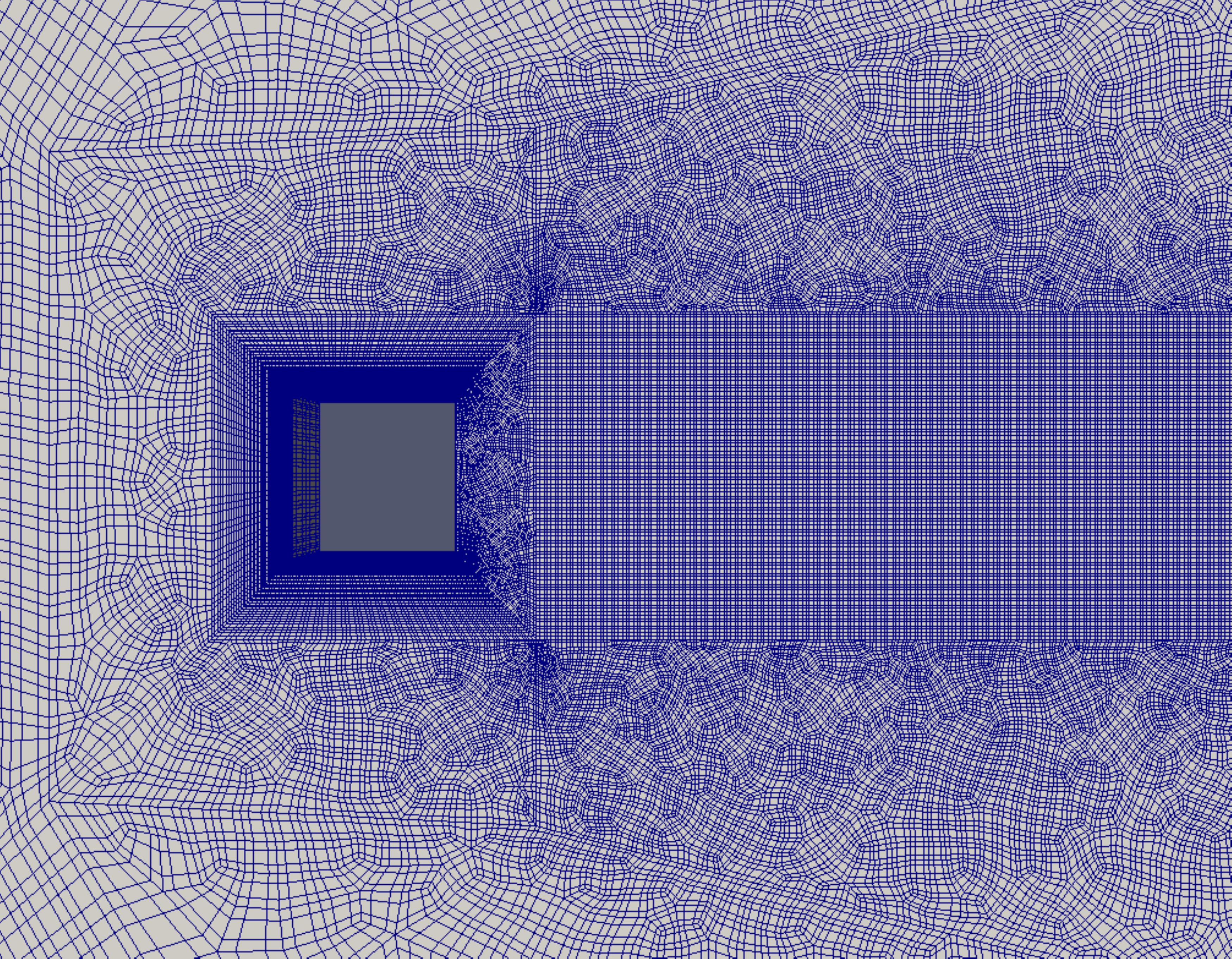}
   \label{fig:subfig1}
   }\\
 \subfigure[Mach number; $\Delta M = 0.0146$.]{
  \includegraphics[trim=0.5cm 1.8cm 0cm 3.6cm, clip=true, totalheight=0.28\textheight]
  {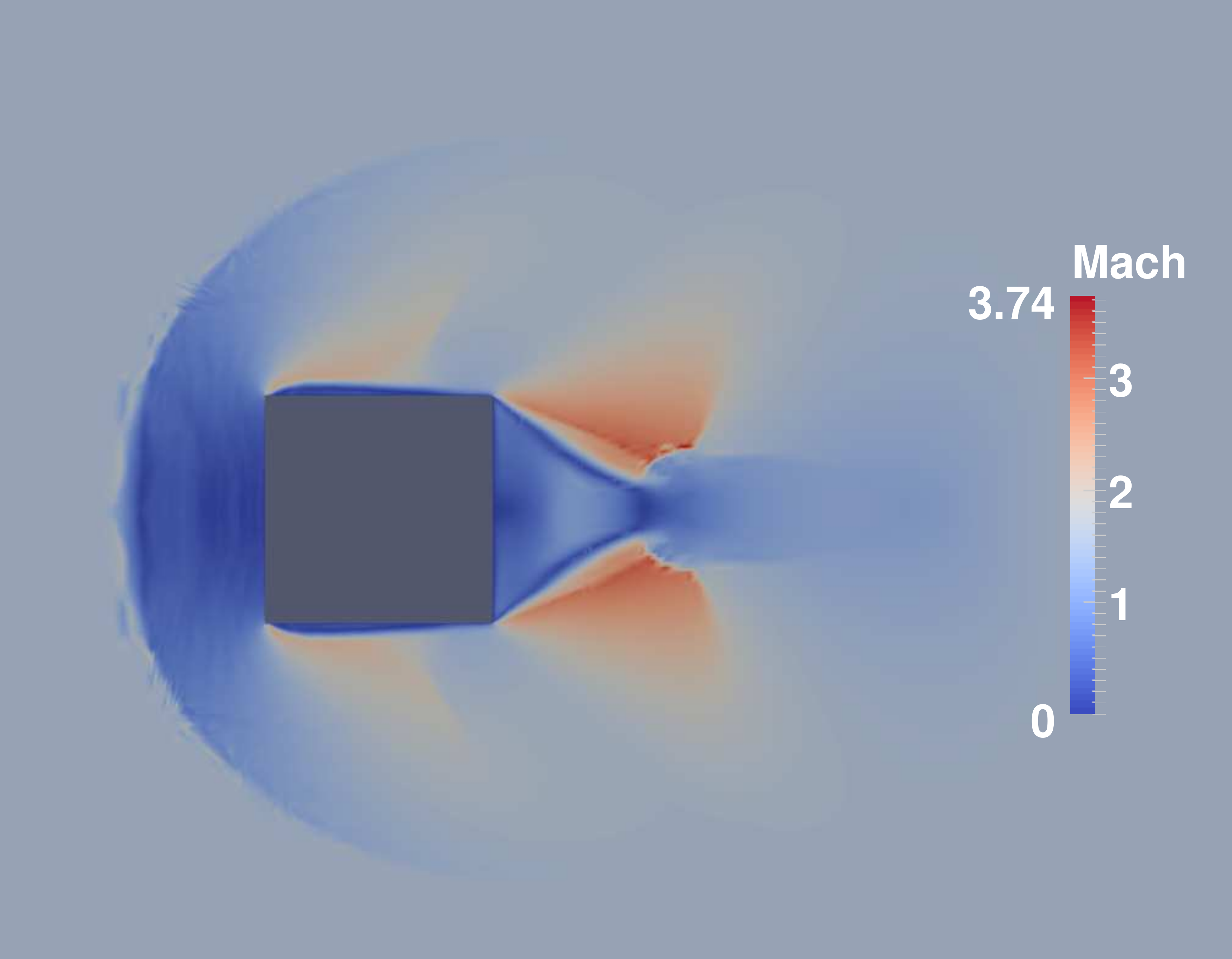}
   \label{fig:subfig2}
   }\\
 \subfigure[Density; $\Delta \rho = 0.0114$.]{
  \includegraphics[trim=1cm 1.8cm 0cm 3.6cm, clip=true, totalheight=0.28\textheight]
  {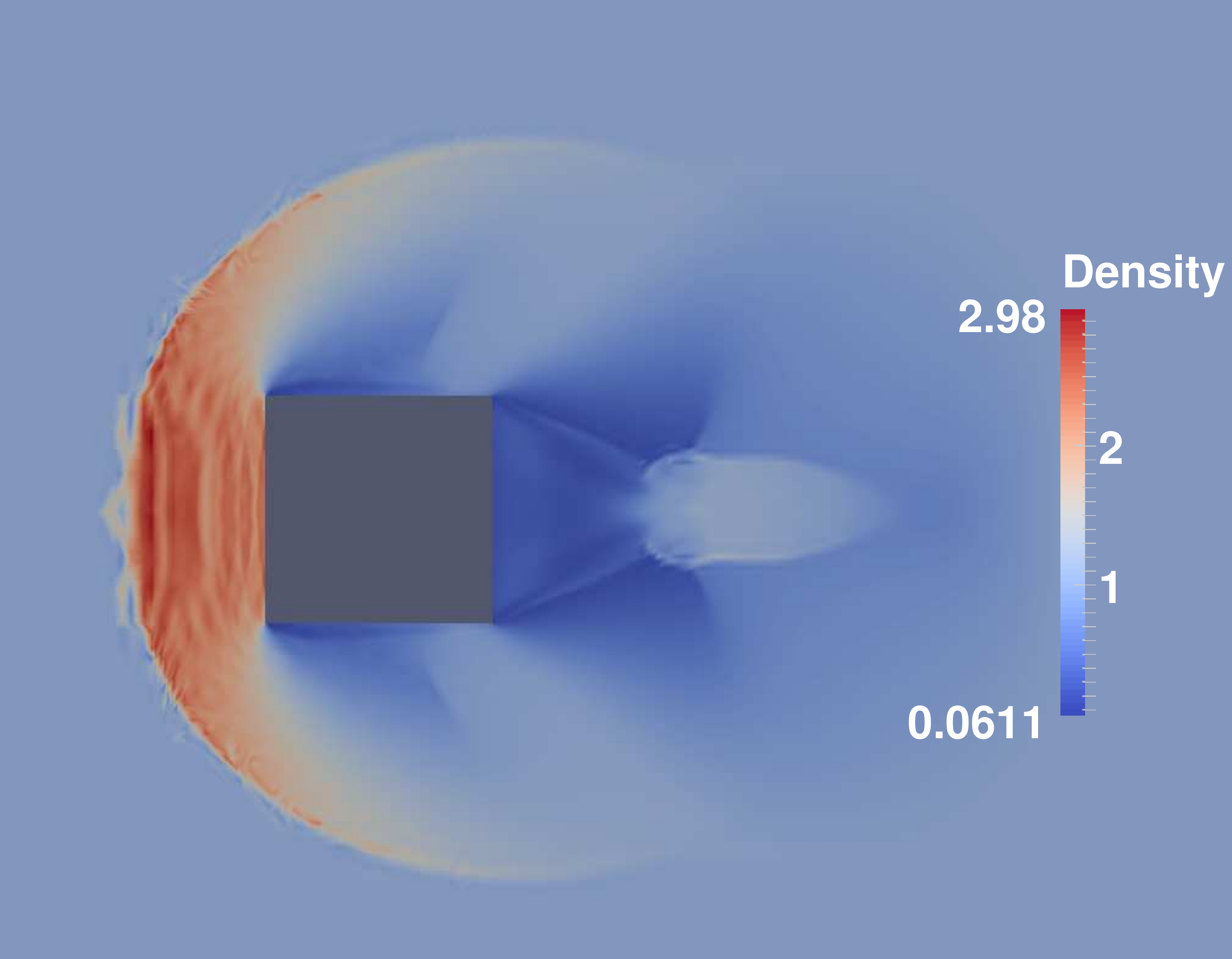}
   \label{fig:subfig3}
   }
 \caption[Optional caption for list of figures]{Unsteady flow past a 3D square cylinder at
   $Re_{\infty}=10^4$ and $M_{\infty}=1.5$; fourth-order ($p=3$) accurate entropy stable
 spatial discretization without stabilization technique; $t=1.5$. \label{fig:supersonic-square-1}}
\end{figure}

In Figures \ref{fig:supersonic-square-3}  
a global view of the ``high order grid'', the Mach 
number, density, temperature and entropy contours at $t=100$ are
shown. At $t=100$, the shock
has already reached a stationary position, and the flow past the square cylinder
is completely unsteady, characterized by subsonic and supersonic regions.
The formation of shock-lets in the near wake region are clearly visible.
\begin{figure}[htbp]
 \centering
 \subfigure[High order grid.]{
  \includegraphics[totalheight=0.26\textheight]
  {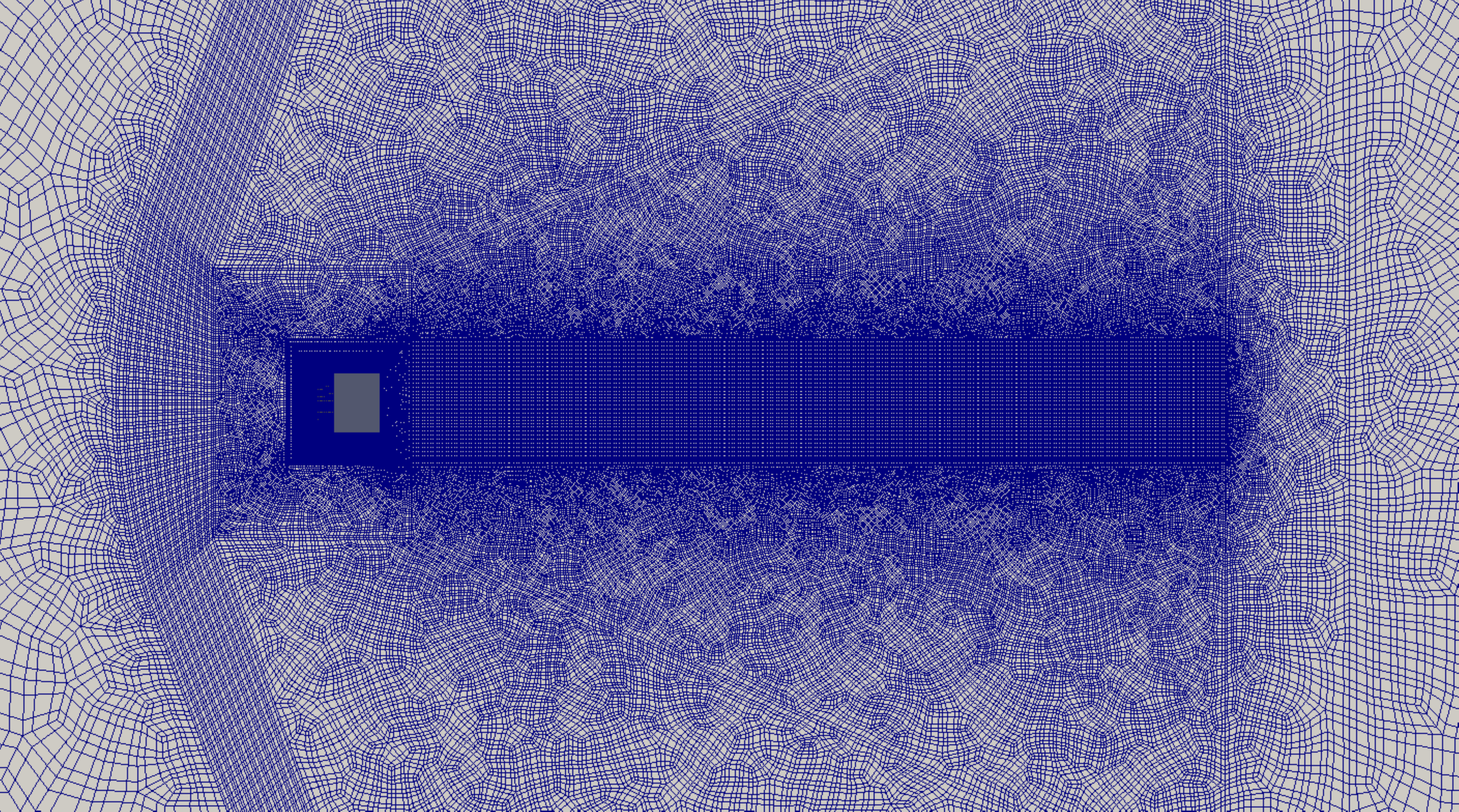}
   \label{fig:subfig4}
   }
 \subfigure[Mach number; $\Delta M = 0.0095$.]{
  \includegraphics[trim=0.5cm 3.5cm 0cm 3.5cm, clip=true,totalheight=0.285\textheight]
  {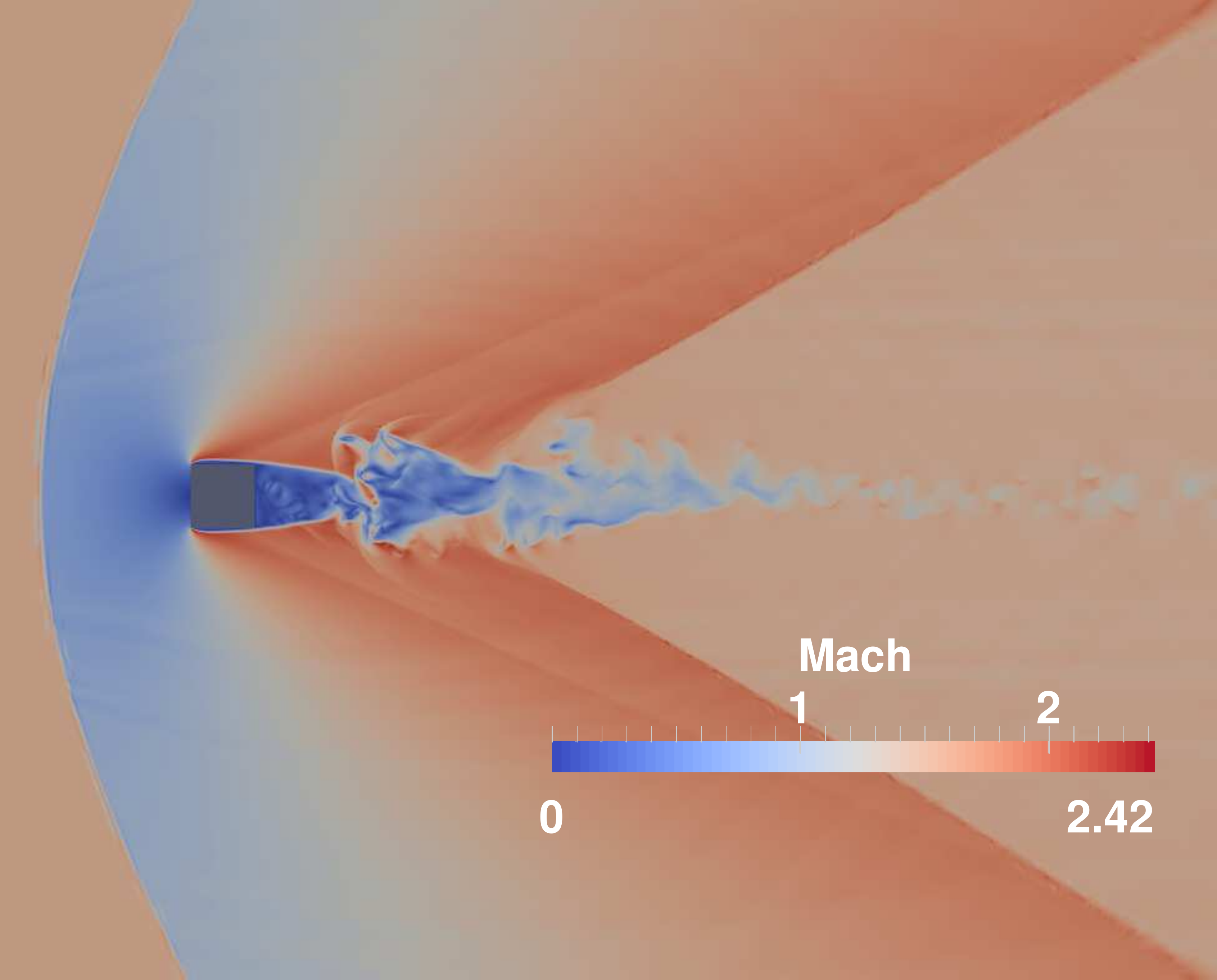}
   \label{fig:subfig5}
   }
 \subfigure[Density; $\Delta \rho = 0.0090$.]{
  \includegraphics[trim=0.5cm 3.5cm 0cm 3.5cm, clip=true,totalheight=0.285\textheight]
  {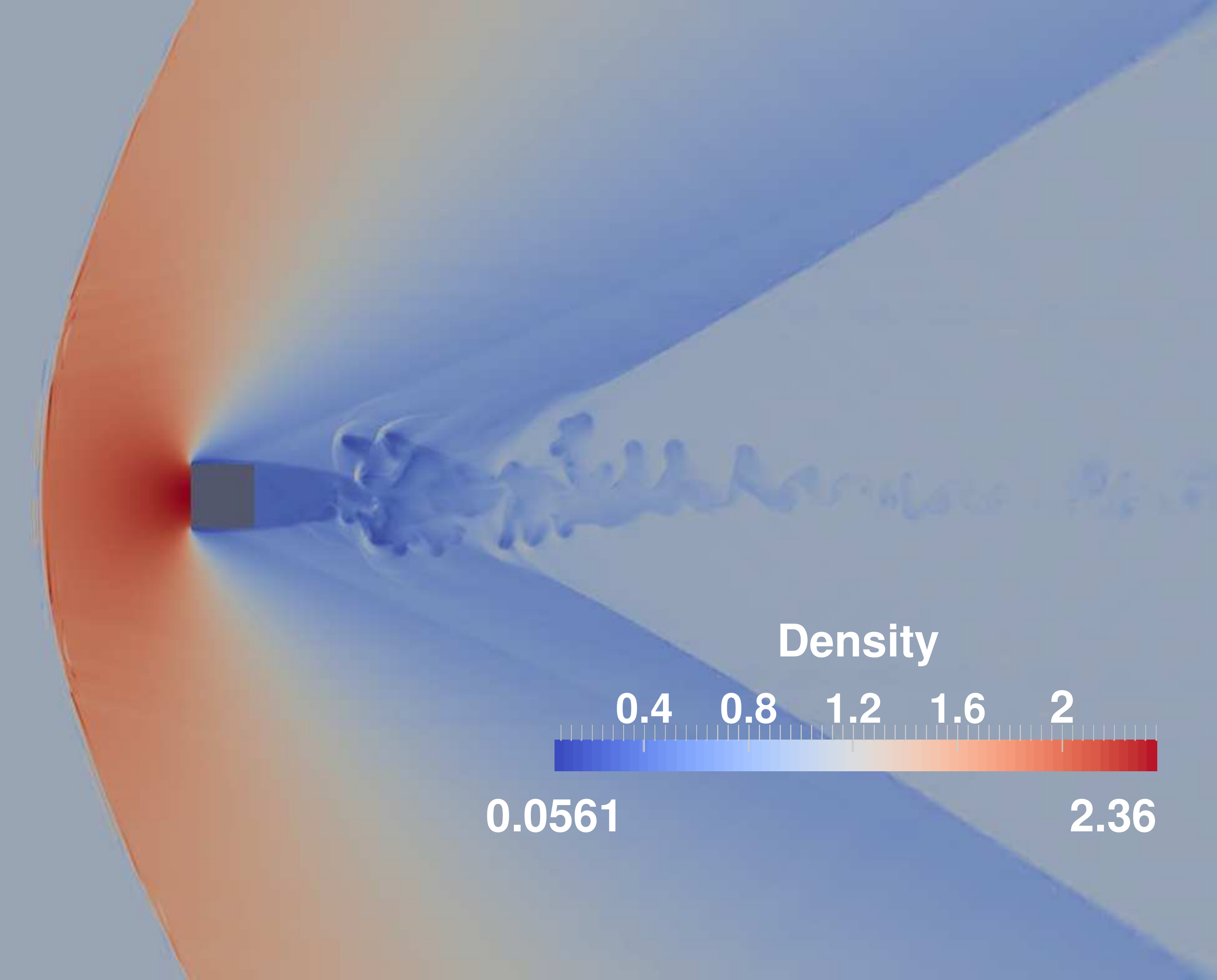}
   \label{fig:subfig6}
   }
\end{figure}
 
\begin{figure}[htbp]
 \centering
 \subfigure[Temperature; $\Delta T = 0.0077$.]{
  \includegraphics[trim=0.5cm 3.5cm 0cm 3.5cm, clip=true,totalheight=0.285\textheight]
  {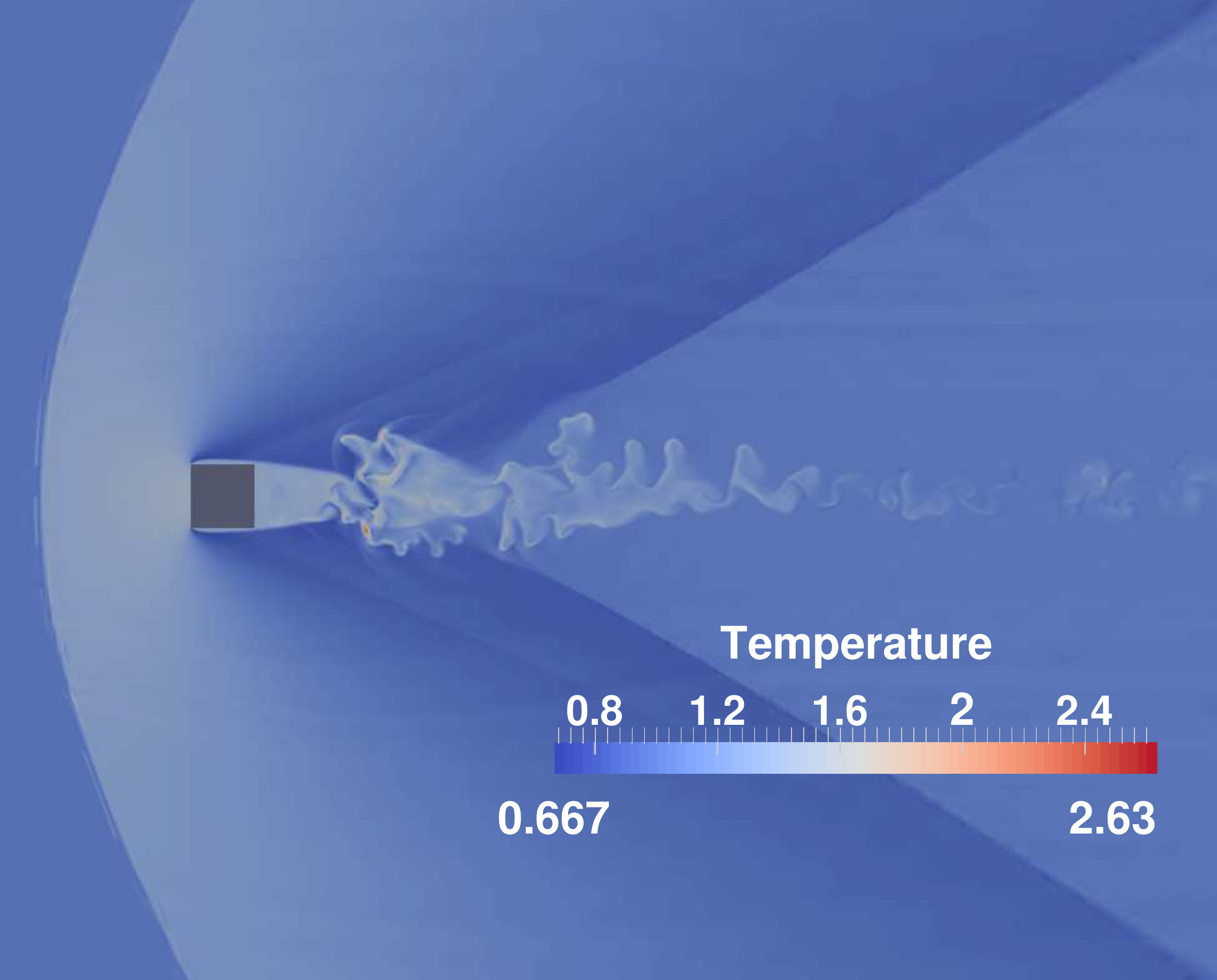}
   \label{fig:subfig7}
   }
 \subfigure[Entropy; $\Delta s = 0.0044$.]{
  \includegraphics[trim=0.5cm 3.5cm 0cm 3.5cm, clip=true,totalheight=0.285\textheight]
    {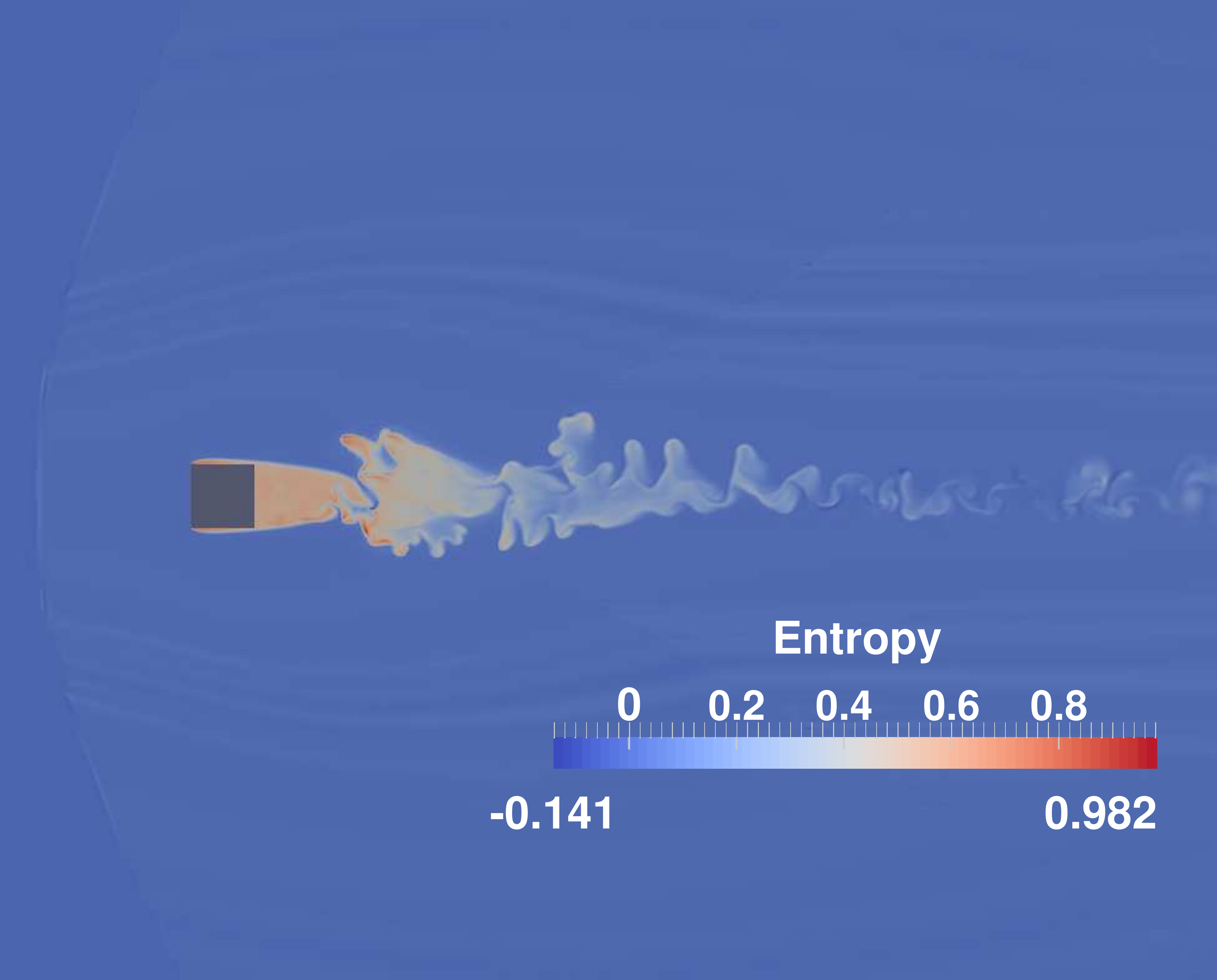}
   \label{fig:subfig8}
   }
 \caption[Optional caption for list of figures]{Unsteady flow past a 3D square cylinder at
   $Re_{\infty}=10^4$ and $M_{\infty}=1.5$; fourth-order ($p=3$) accurate entropy stable
 spatial discretization without stabilization technique; $t=100$. \label{fig:supersonic-square-3}}
\end{figure}


\section{Conclusions}\label{sec:conclusions}
Herein, we have shown that no-slip boundary conditions together with
a boundary condition on the heat entropy flow, $\left(1/T \, \partial T/\partial n \right)_{wall}$, 
imply stability for the continuous compressible Navier--Stokes equations. The
boundary condition on the heat entropy flow is in complete agreement with the 
thermodynamic (entropy) analysis of a generic system.  
An entropy stable numerical procedure is presented for weakly enforcing 
these solid wall boundary conditions via a penalty approach.
The resulting semi-discrete operator mimics exactly the behavior at the continuous level. 
The proposed non-linear boundary treatment provides a mechanism for ensuring the non-linear stability in 
the $L^2$ norm of the continuous and semi-discretized compressible Navier--Stokes equations. 
Although discontinuous spectral collocation operators are used in this work, the 
new boundary conditions are 
compatible with any diagonal norm summation-by-parts spatial operator, including 
finite element, finite volume, finite difference, discontinuous Galerkin, and flux reconstruction
schemes.

Numerical computations around a three-dimensional square cylinder in the subsonic regime are
performed to highlight the accuracy and robustness of the proposed numerical
procedure. Measurement of forces on the cylinder showed very good agreement 
with the results available from the literature. Furthermore, we have shown that
the no-slip conditions approach zero to design-order (i.e., the convergence rate
is $p+1$), and the heat entropy flow converges to the prescribed boundary value
at a rate of $p$. 
 
The robustness of the complete semi-discrete operator 
(i.e., the entropy stable interior operator coupled with the new boundary 
treatment) has been demonstrated for the supersonic flow past a
three-dimensional square 
cylinder at $Re_{\infty}=10\time10^4$ and $M_{\infty} = 1.5$. This test has 
been successfully computed with a fourth-order accurate method without the 
need to introduce artificial dissipation, limiting techniques or filtering, 
for stabilizing the computations, a feat unattainable with several alternative 
approaches just based on linear analysis.

This work clearly indicates that, although incremental improvements to existing 
algorithms will continue to improve overall capabilities, the development of 
novel robust numerical techniques such as entropy preserving or entropy stable
schemes and their extension to complex multi-scale and multi-physics problems offers 
the possibility of radical advances in computational fluid dynamics and
computational aerodynamics in terms of robustness, fidelity
and efficiency.

\section*{Acknowledgments}
Special thanks are extended to Dr. Mujeeb Malik for funding this work as part of
the ``Revolutionary Computational Aerosciences'' project. This research was also 
supported by an appointment to the NASA Postdoctoral 
Program at the Langley Research Center, administered by Oak Ridge Associated 
Universities through a contract with NASA. The authors are also
grateful to Professor Magnus Sv\"{a}rd for the fruitful discussions on entropy stability.

\appendix\label{appendix}

\section{Coefficient matrices of the viscous flux}\label{app:cij}
The viscous coefficient matrices $c_{ij}^{\prime}$ used to define the viscous 
fluxes in Cartesian coordinates in \eqref{eq:visc-flux-cij} are defined as
\begin{equation*}
  c_{11}^{\prime} =
\begin{pmatrix}
  0 & 0 & 0 & 0 & 0 \\
  0 & \frac{4}{3}\mu & 0 & 0 & 0  \\
  0 & 0 & \mu & 0 & 0 \\
  0 & 0 & 0 & \mu & 0 \\
  0 & \frac{4}{3}\mu u_1 & \mu u_2 & \mu u_3 & \kappa
\end{pmatrix},
\quad
  c_{12}^{\prime} =
\begin{pmatrix}
  0 & 0 & 0 & 0 & 0 \\
  0 & 0 & -\frac{2}{3}\mu & 0 & 0  \\
  0 & \mu & 0 & 0 & 0 \\
  0 & 0 & 0 & 0 & 0 \\
  0 & \mu u_2 & -\frac{2}{3}\mu u_1 & 0 & 0
\end{pmatrix},
\end{equation*}
\begin{equation*}
  c_{13}^{\prime} =
\begin{pmatrix}
  0 & 0 & 0 & 0 & 0 \\
  0 & 0 & 0 & -\frac{2}{3}\mu & 0  \\
  0 & 0 & 0 & 0 & 0 \\
  0 & \mu & 0 & 0 & 0 \\
  0 & \mu u_3 & 0 & -\frac{2}{3}\mu u_1 & 0
\end{pmatrix},
\quad
  c_{21}^{\prime} =
\begin{pmatrix}
  0 & 0 & 0 & 0 & 0 \\
  0 & 0 & \mu & 0 & 0  \\
  0 & -\frac{2}{3}\mu & 0 & 0 & 0 \\
  0 & 0 & 0 & 0 & 0 \\
  0 & -\frac{2}{3}\mu u_2 & \mu u_1 & 0 & 0
\end{pmatrix},
\end{equation*}
\begin{equation*}
  c_{22}^{\prime} =
\begin{pmatrix}
  0 & 0 & 0 & 0 & 0 \\
  0 & \mu & 0 & 0 & 0  \\
  0 & 0 & \frac{4}{3}\mu & 0 & 0 \\
  0 & 0 & 0 & \mu & 0 \\
  0 & \mu u_1 & \frac{4}{3}\mu u_2 & \mu u_3 & \kappa
\end{pmatrix},
\quad
  c_{23}^{\prime} =
\begin{pmatrix}
  0 & 0 & 0 & 0 & 0 \\
  0 & 0 & 0 & 0 & 0  \\
  0 & 0 & 0 & -\frac{2}{3}\mu & 0 \\
  0 & 0 & \mu & 0 & 0 \\
  0 & 0 & \mu u_3 & -\frac{2}{3}\mu u_2 & 0
\end{pmatrix},
\end{equation*}
\begin{equation*}
  c_{31}^{\prime} =
\begin{pmatrix}
  0 & 0 & 0 & 0 & 0 \\
  0 & 0 & 0 & \mu & 0  \\
  0 & 0 & 0 & 0 & 0 \\
  0 & -\frac{2}{3}\mu & 0 & 0 & 0 \\
  0 & -\frac{2}{3}\mu u_3 & 0 & \mu u_1 & 0
\end{pmatrix},
\quad
  c_{32}^{\prime} =
\begin{pmatrix}
  0 & 0 & 0 & 0 & 0 \\
  0 & 0 & 0 & 0 & 0  \\
  0 & 0 & 0 & \mu & 0 \\
  0 & 0 & -\frac{2}{3}\mu & 0 & 0 \\
  0 & 0 & -\frac{2}{3}\mu u_3 & \mu u_2 & 0
\end{pmatrix},
\end{equation*}
\begin{equation*}
  c_{33}^{\prime} =
\begin{pmatrix}
  0 & 0 & 0 & 0 & 0 \\
  0 & \mu & 0 & 0 & 0  \\
  0 & 0 & \mu & 0 & 0 \\
  0 & 0 & 0 & \frac{4}{3}\mu & 0 \\
  0 & \mu u_1 & \mu u_2 & \frac{4}{3}\mu u_3 & \kappa
\end{pmatrix}.
\end{equation*}

The symmetrized coefficient matrices used in \eqref{eq:entropyChainRule} to
define the viscous fluxes as a function of the gradient of the entropy variables
are found using
\begin{equation*}
  \widehat{c}_{ij} = c_{ij} \frac{\partial q}{\partial w} = c_{ij}^{\prime}
  \frac{\partial v}{\partial w}.
\end{equation*}
Therefore, they take the following form:
\begin{equation*}
  \widehat{c}_{11} =
\begin{pmatrix}
  0 & 0 & 0 & 0 & 0 \\
  0 & \frac{4}{3}T\mu & 0 & 0 & \frac{4}{3}T\mu u_1  \\
  0 & 0 & T\mu & 0 & T\mu u_2 \\
  0 & 0 & 0 & T\mu & T \mu u_3 \\
  0 & \frac{4}{3}T\mu u_1 & T\mu u_2 & T\mu u_3 & T^2\kappa +
  \frac{1}{3}T\left(4\mu u_1^2 + 3 \mu u_2^2 + 3 \mu u_3^2\right)
\end{pmatrix},
\end{equation*}
\begin{equation*}
\widehat{c}_{22} =
\begin{pmatrix}
  0 & 0 & 0 & 0 & 0 \\
  0 & T\mu & 0 & 0 & T\mu u_1  \\
  0 & 0 & \frac{4}{3}T\mu & 0 & \frac{4}{3}T\mu u_2 \\
  0 & 0 & 0 & T\mu & T \mu u_3 \\
  0 & T\mu u_1 & \frac{4}{3}T\mu u_2 & T\mu u_3 & T^2\kappa +
  \frac{1}{3}T\left(3\mu u_1^2 + 4 \mu u_2^2 + 3 \mu u_3^2\right)
\end{pmatrix},
\end{equation*}
\begin{equation*}
  \widehat{c}_{33} =
\begin{pmatrix}
  0 & 0 & 0 & 0 & 0 \\
  0 & T\mu & 0 & 0 & T\mu u_1  \\
  0 & 0 & T\mu & 0 & T\mu u_2 \\
  0 & 0 & 0 & \frac{4}{3}T\mu & \frac{4}{3}T \mu u_3 \\
  0 & T\mu u_1 & T\mu u_2 & \frac{4}{3}T\mu u_3 & T^2\kappa +
  \frac{1}{3}T\left(3\mu u_1^2 + 3 \mu u_2^2 + 4 \mu u_3^2\right)
\end{pmatrix},
\end{equation*}
\begin{equation*}
  \widehat{c}_{12} =
\begin{pmatrix}
  0 & 0 & 0 & 0 & 0 \\
  0 & 0 & -\frac{2}{3}T\mu & 0 & -\frac{2}{3}T\mu u_2  \\
  0 & T\mu & 0 & 0 & T\mu u_1 \\
  0 & 0 & 0 & 0 & 0 \\
  0 & T\mu u_2 & -\frac{2}{3}T\mu u_1 & 0 & \frac{1}{3}T\mu u_1 u_2
\end{pmatrix},
\quad
  \widehat{c}_{13} =
\begin{pmatrix}
  0 & 0 & 0 & 0 & 0 \\
  0 & 0 & 0 & -\frac{2}{3}T\mu & -\frac{2}{3}T\mu u_3  \\
  0 & 0 & 0 & 0 & 0 \\
  0 & T\mu & 0 & 0 & T\mu u_1 \\
  0 & T\mu u_3 & 0 & -\frac{2}{3}T\mu u_1 & \frac{1}{3}T\mu u_1 u_3
\end{pmatrix},
\end{equation*}
\begin{equation*}
  \widehat{c}_{23} =
\begin{pmatrix}
  0 & 0 & 0 & 0 & 0 \\
  0 & 0 & 0 & 0 & 0  \\
  0 & 0 & 0 & -\frac{2}{3}T\mu & -\frac{2}{3}T\mu u_3 \\
  0 & 0 & T\mu & 0 & T\mu u_2 \\
  0 & 0 & T\mu u_3 & -\frac{2}{3}T\mu u_2 & \frac{1}{3}T\mu u_2 u_3
\end{pmatrix},
\end{equation*}
where $$
\widehat{c}_{21} = \widehat{c}_{12}^{\top}, \quad 
\widehat{c}_{31} = \widehat{c}_{13}^{\top}, \quad
\widehat{c}_{32} = \widehat{c}_{23}^{\top}.
$$

\section{Counter example of non-linear wall boundary
conditions}\label{app:counter-example}

Carrying out the entropy stability analysis by using expression
\eqref{eq:SAT-no-slip-bc-wrong}, it can be shown that the inviscid
penalty in \eqref{eq:SAT-no-slip-bc-wrong}
is entropy conservative (see the proof of Theorem 
\ref{th:euler-flipping-sign}). Therefore, the remaining relation to analyze is
\begin{equation}\label{eq:estimate-no-slip-bc-viscous-wrong}
  \begin{aligned}
    \frac{d}{d t} \left\|{S}\right\|^2_{\pmatv} \:+\:  \wN^{\top}
    \pmat_{x_2 x_3} \, \gmat_{(1)} 
    \fM_1^{(V)} +  \mathbf{DT}  
    \leq &+ \wN^{\top}
    \pmat_{x_2 x_3} \, \gmat_{(1)} [L] \left(\wN-\gb^{(NS)}\right).
\end{aligned}
\end{equation}
To obtain a quadratic form in boundary terms we need to borrow from 
$\mathbf{DT}$:
\begin{equation}
\begin{aligned}
\mathbf{DT} & = \sum_{i=1}^3 \sum_{j=1}^3 \left(\dmat_{x_i} \wN\right)^{\top}
\pmatv [\chatmat_{ij}]  \left(\dmat_{x_j} \wN\right) \\
& =
\begin{pmatrix}
  \dmat_{x_1} \, \wN \\ \dmat_{x_2} \, \wN  \\ \dmat_{x_3} \, \wN
\end{pmatrix}^{\top}
\begin{pmatrix}
  \pmatv[\chatmat_{11}] & \pmatv[\chatmat_{12}] & \pmatv[\chatmat_{13}] \\
  \pmatv[\chatmat_{21}] & \pmatv[\chatmat_{22}] & \pmatv[\chatmat_{23}] \\
  \pmatv[\chatmat_{31}] & \pmatv[\chatmat_{32}] & \pmatv[\chatmat_{33}] \\
\end{pmatrix}
\begin{pmatrix}
  \dmat_{x_1} \, \wN \\ \dmat_{x_2} \, \wN  \\
  \dmat_{x_3} \, \wN
\end{pmatrix} \\
& =
\begin{pmatrix}
  \dmat_{x_1} \, \wN \\ \dmat_{x_2} \, \wN  \\ \dmat_{x_3} \, \wN
\end{pmatrix}^{\top}
\left(\pmat_{x_1}\right)_{(1)(1)} 
\begin{pmatrix}
  \pmat_{x_2 x_3}[\chatmat_{11}] & \pmat_{x_2 x_3}[\chatmat_{12}] & \pmat_{x_2 x_3}[\chatmat_{13}] \\
  \pmat_{x_2 x_3}[\chatmat_{21}] & \pmat_{x_2 x_3}[\chatmat_{22}] & \pmat_{x_2 x_3}[\chatmat_{23}] \\
  \pmat_{x_2 x_3}[\chatmat_{31}] & \pmat_{x_2 x_3}[\chatmat_{32}] & \pmat_{x_2 x_3}[\chatmat_{33}] \\
\end{pmatrix}
\begin{pmatrix}
  \dmat_{x_1} \, \wN \\ \dmat_{x_2} \, \wN  \\ \dmat_{x_3} \, \wN
\end{pmatrix}
+ \widetilde{\mathbf{DT}} \\
& =
\begin{pmatrix}
  \dmat_{x_1} \, \wN \\ \dmat_{x_2} \, \wN  \\ \dmat_{x_3} \, \wN
\end{pmatrix}^{\top}
\left(\pmat_{x_1}\right)_{(1)(1)}
\pmat^{\prime}
\begin{pmatrix}
  [\chatmat_{11}] & [\chatmat_{12}] & [\chatmat_{13}] \\
  [\chatmat_{21}] & [\chatmat_{22}] & [\chatmat_{23}] \\
  [\chatmat_{31}] & [\chatmat_{32}] & [\chatmat_{33}] \\
\end{pmatrix}
\begin{pmatrix}
  \dmat_{x_1} \, \wN \\ \dmat_{x_2} \, \wN  \\ \dmat_{x_3} \, \wN
\end{pmatrix}
+ \widetilde{\mathbf{DT}},
\end{aligned}
\end{equation}
where $$\pmat^{\prime} =
\textrm{diag}\left(\pmat_{x_2 x_3}, \pmat_{x_2 x_3},
\pmat_{x_2 x_3}\right),$$ and the scalar 
$\left(\pmat_{x_1}\right)_{(1)(1)} > 0$. 
Therefore, Equation 
\eqref{eq:estimate-no-slip-bc-viscous-wrong} may be written as
\begin{equation}\label{eq:entropy-estimate-no-slip-wrong}
\begin{aligned}
   \frac{d}{d t} & \left\|{S}\right\|^2_{\pmatv}  + \widetilde{\mathbf{DT}} \\
\leq 
& \:+\:  \frac{1}{2} \, \widetilde{\wN}^{\top} 
\widetilde{\pmat}^{\prime}
\begin{pmatrix}
  [L]          &  -[\chatmat_{11}] & -[\chatmat_{12}] & -[\chatmat_{13}] \\
  -[\chatmat_{11}] & -2\left(\pmat_{x_1}\right)_{(1)(1)} \, [\chatmat_{11}] &
  -2\left(\pmat_{x_1}\right)_{(1)(1)}
  [\chatmat_{12}] & -2\left(\pmat_{x_1}\right)_{(1)(1)} \, [\chatmat_{13}] \\
  -[\chatmat_{12}] & -2\left(\pmat_{x_1}\right)_{(1)(1)} \, [\chatmat_{12}] &
  -2\left(\pmat_{x_1}\right)_{(1)(1)}
  [\chatmat_{22}] & -2\left(\pmat_{x_1}\right)_{(1)(1)} \, [\chatmat_{23}] \\
  -[\chatmat_{13]} & -2\left(\pmat_{x_1}\right)_{(1)(1)} \, [\chatmat_{13}] &
  -2\left(\pmat_{x_1}\right)_{(1)(1)}
  [\chatmat_{23}] & -2\left(\pmat_{x_1}\right)_{(1)(1)} \, [\chatmat_{33}]
\end{pmatrix}
 \widetilde{\wN} \\
 & + \frac{1}{2} \wN^{\top} \pmat_{x_2 x_3} \, \gmat_{(1)}
[L] \, \wN - \wN^{\top}
\pmat_{x_2 x_3} \, \gmat_{(1)} [L] \, \gb^{(NS)},
\end{aligned}
\end{equation}
where $$\widetilde{\pmat}^{\prime} =
\textrm{diag}\left(\pmat_{x_2 x_3}, \pmat_{x_2 x_3},
\pmat_{x_2 x_3},\pmat_{x_2 x_3}\right),$$ and 
\begin{equation}\label{eq:wtilde}
  \widetilde{\wN}^{\top} = \left((\wN_{(0,x_2,x_3)}^{\top}, \mathbf{0}, \left(\dmat_{x_1} \,
    \wN\right)_{(0,x_2,x_3)}^{\top}, \left(\dmat_{x_2} \, \wN\right)_{(0,x_2,x_3)}^{\top},
  \left(\dmat_{x_3} \, \wN\right)_{(0,x_2,x_3)}^{\top}, \mathbf{0} \right).
\end{equation}
The bold zeros in \eqref{eq:wtilde} indicates that all the numerical states
(entropy variables and gradients of the entropy variables) of the nodes which do not lie 
on the plane $(0,x_2,x_3)$
 are set to zero.

To bound the time derivative of the entropy we must ensure that
each term in \eqref{eq:entropy-estimate-no-slip-wrong} is bounded.
The first contribution on the RHS is a quadratic term in 
$\widetilde{\wN}$ and dissipative if 
the large matrix  
\begin{equation}\label{eq:large-wrong}
 \frac{1}{2}
\begin{pmatrix}
  [L]          &  -[\chatmat_{11}] & -[\chatmat_{12}] & -[\chatmat_{13}] \\
  -[\chatmat_{11}] & -2\left(\pmat_{x_1}\right)_{(1)(1)} \, [\chatmat_{11}] &
  -2\left(\pmat_{x_1}\right)_{(1)(1)}
  [\chatmat_{12}] & -2\left(\pmat_{x_1}\right)_{(1)(1)} \, [\chatmat_{13}] \\
  -[\chatmat_{12}] & -2\left(\pmat_{x_1}\right)_{(1)(1)} \, [\chatmat_{12}] &
  -2\left(\pmat_{x_1}\right)_{(1)(1)}
  [\chatmat_{22}] & -2\left(\pmat_{x_1}\right)_{(1)(1)} \, [\chatmat_{23}] \\
  -[\chatmat_{13]} & -2\left(\pmat_{x_1}\right)_{(1)(1)} \, [\chatmat_{13}] &
  -2\left(\pmat_{x_1}\right)_{(1)(1)}
  [\chatmat_{23}] & -2\left(\pmat_{x_1}\right)_{(1)(1)} \, [\chatmat_{33}]
\end{pmatrix}
\end{equation}
is symmetric negative semi-definite. However, to ensure that we only need to construct  
the following $20 \times 20$ matrix,
\begin{equation}\label{eq:gamma-mat}
  \Gamma =\frac{1}{2}
\begin{pmatrix}
  L          & -\widehat{c}_{11} & -\widehat{c}_{12} & -\widehat{c}_{13} \\
  -\widehat{c}_{11} & -2\left(\pmat_{x_1}\right)_{(1)(1)} \, \widehat{c}_{11}
  & -2\left(\pmat_{x_1}\right)_{(1)(1)}
  \widehat{c}_{12} & -2\left(\pmat_{x_1}\right)_{(1)(1)} \, \widehat{c}_{13} \\
  -\widehat{c}_{12} & -2\left(\pmat_{x_1}\right)_{(1)(1)} \, \widehat{c}_{12} & 
  -2\left(\pmat_{x_1}\right)_{(1)(1)}
  \widehat{c}_{22} & -2\left(\pmat_{x_1}\right)_{(1)(1)} \, \widehat{c}_{23} \\
  -\widehat{c}_{13} & -2\left(\pmat_{x_1}\right)_{(1)(1)} \, \widehat{c}_{13} & 
  -2\left(\pmat_{x_1}\right)_{(1)(1)}
  \widehat{c}_{23} & -2\left(\pmat_{x_1}\right)_{(1)(1)} \, \widehat{c}_{33}
\end{pmatrix},
\end{equation}
so that it is symmetric negative semi-definite.
The matrices $\widehat{c}_{ij}, i,j=1,2,3$, in \eqref{eq:gamma-mat} are constructed using the
primitive variable at the usual boundary node.
The rows and columns of the matrix $\Gamma$ corresponding 
to the density components are all zero because the first component of the 
viscous fluxes is zero. Therefore, such 
rows and columns do not affect the negativity of 
\eqref{eq:gamma-mat}. The matrix $\Gamma$ can be expressed in block form as
\begin{equation}\label{eq:block-form-matrix-gamma}
  \Gamma = 
\begin{pmatrix}
  A & B \\
  B^{\top} & D
\end{pmatrix}.
\end{equation}
The condition on the five-by-five matrix $L$ that ensures the negative-definiteness of 
\eqref{eq:block-form-matrix-gamma} can obtained by requiring that the Schur
complement of $\Gamma$ is negative,
\begin{equation}\label{eq:schur-no-slip-bc}
 Schur = A - B \, D^{-1}\, B^{\top} = -\frac{\widehat{c}_{11}+2\, L \,
\left(\pmat_{x_1}\right)_{(1)(1)}}{4\left(\pmat_{x_1}\right)_{(1)(1)}} < 0.
\end{equation}
The inequality in \eqref{eq:schur-no-slip-bc} is a sufficient condition because 
the block matrix $D$ is already well behaved (i.e., it is already symmetric and 
positive semi-definite). Thus, the Schur complement 
\eqref{eq:schur-no-slip-bc} is smaller than or equal to zero if
\begin{equation}\label{eq:condition-M-no-slip-wrong}
  L \leq - \frac{\widehat{c}_{11}}{2\,\left(\pmat_{x_1}\right)_{(1)(1)}}.
\end{equation}

The last two terms on the RHS of expression
\eqref{eq:entropy-estimate-no-slip-wrong} are the remaining contributions to
bound. Such terms can be re-written in a quadratic form 
as
\begin{equation}\label{eq:wrong-IP}
  \begin{aligned}
    \frac{1}{2}\wN^{\top} \pmat_{x_2 x_3} \, \gmat_{(1)}
[L] \, \wN - \wN^{\top}
\pmat_{x_2 x_3} \, \gmat_{(1)} [L] \, \gb^{(NS)} = & + \frac{1}{2}
 \left(\wN - \gb^{(NS)}\right)^{\top} \pmat_{x_2 x_3}\, 
 \gmat_{(1)} \, 
 [L] \left(\wN - \gb^{(NS)}\right) \\
 & -\frac{1}{2}
\left(\gb^{(NS)}\right)^{\top} \pmat_{x_2
 x_3} \, \gmat_{(1)} \, [L] \, \gb^{(NS)}.
 \end{aligned}
 \end{equation}
 From inequality \eqref{eq:condition-M-no-slip-wrong}, we know that the matrix
 $[L]$ is 
 negative definite or negative semi-definite. Therefore, the first term, which is quadratic in 
 $\left(\wN - \gb^{(NS)}\right)$, is dissipative. However, the second contribution is a 
 positive term and cannot be bounded because it is not only a function of the imposed 
 boundary data $\gb^{(NS)}$. In fact, the element in the fifth row and fifth column of the
 matrix $L$ is non-zero and it is a function of the numerical solution through 
 relation \eqref{eq:condition-M-no-slip-wrong} (i.e., through the matrix 
 $\widehat{c}_{11}$ which is built from the numerical state at the boundary node).

\bibliographystyle{elsarticle-num} 
\bibliography{entropy-stable-solid-wall-bc}

\end{document}